\documentclass[journal,12pt,onecolumn,draftclsnofoot,]{IEEEtran}
\IEEEoverridecommandlockouts
\usepackage{cite}
\usepackage{amsmath,amssymb,amsfonts}
\usepackage{algorithmic}
\usepackage{graphicx}
\usepackage{textcomp}
\usepackage{multirow, tabu}
\usepackage[table,xcdraw]{xcolor}
\usepackage{enumitem}
\usepackage{array}
\usepackage{colortbl}
\graphicspath{{./images/}}
\usepackage[mathscr]{euscript}
\usepackage{stfloats}
\usepackage{rsfso}
\usepackage{color,soul}
\usepackage{xcolor}
\usepackage{amsthm}
\usepackage{subfig}

\usepackage{scalerel}
\usepackage{tikz}
\usetikzlibrary{svg.path}
\definecolor{orcidlogocol}{HTML}{A6CE39}
\tikzset{
  orcidlogo/.pic={
    \fill[orcidlogocol] svg{M256,128c0,70.7-57.3,128-128,128C57.3,256,0,198.7,0,128C0,57.3,57.3,0,128,0C198.7,0,256,57.3,256,128z};
    \fill[white] svg{M86.3,186.2H70.9V79.1h15.4v48.4V186.2z}
                 svg{M108.9,79.1h41.6c39.6,0,57,28.3,57,53.6c0,27.5-21.5,53.6-56.8,53.6h-41.8V79.1z M124.3,172.4h24.5c34.9,0,42.9-26.5,42.9-39.7c0-21.5-13.7-39.7-43.7-39.7h-23.7V172.4z}
                 svg{M88.7,56.8c0,5.5-4.5,10.1-10.1,10.1c-5.6,0-10.1-4.6-10.1-10.1c0-5.6,4.5-10.1,10.1-10.1C84.2,46.7,88.7,51.3,88.7,56.8z};
  }
}
\newcommand\orcidicon[1]{\href{https://orcid.org/#1}{\mbox{\scalerel*{
\begin{tikzpicture}[yscale=-1,transform shape]
\pic{orcidlogo};
\end{tikzpicture}
}{|}}}}
\usepackage{hyperref} 

\newtheorem{thm}{Theorem}
\newtheorem{cor}{Corollary}
\newtheorem{lem}{Lemma}

\newcommand{\BSlash}{\char`\\}

\newcommand{\ie}{i.e.,~}

\def\BibTeX{{\rm B\kern-.05em{\sc i\kern-.025em b}\kern-.08em
    T\kern-.1667em\lower.7ex\hbox{E}\kern-.125emX}}

\begin{document}

\title{\huge Multiple-Association Supporting HTC/MTC in Limited-Backhaul Capacity Ultra-Dense Networks}

\author{

Mohammed Elbayoumi$^{\orcidicon{0000-0002-0147-2144}}$, \IEEEmembership{Student Member,~IEEE}, Walaa Hamouda$^{\orcidicon{0000-0001-6618-5851}}$, \IEEEmembership{Senior Member,~IEEE}, and Amr Youssef$^{\orcidicon{0000-0002-4284-8646}}$, \IEEEmembership{Senior Member,~IEEE}
\thanks{M. Elbayoumi and A. Youssef are with Concordia Institute for Information Systems	Engineering, Concordia University, Montreal, QC H3G 1M8, Canada (e-mail:	mo\_lsay$@$encs.concordia.ca; youssef$@$ciise.concordia.ca).} 
\thanks{W. Hamouda is with the Department of Electrical and Computer Engineering, Concordia University, Montreal, QC H3G 1M8, Canada (e-mail: hamouda$@$ece.concordia.ca).}
}

\maketitle

\begin{abstract}
    Coexistence of Human-Type Communications (HTCs) and Machine-Type Communications (MTCs) is inevitable. Ultra-Dense Networks (UDNs) will be efficacious in supporting both types of communications. In a UDN, a massive number of low-power and low-cost Small Cells (SCs) are deployed with density higher than that of the HTC users. In such a scenario, the backhaul capacities constitute an intrinsic bottleneck for the system. Hence, we propose a multiple association scheme where each HTC user associates to and activates multiple SCs to overcome the backhaul capacity constraints. In addition, having more active cells allows for more MTC devices to be supported by the network. Using tools from stochastic geometry, we formulate a novel mathematical framework investigating the performance of the limited-backhaul capacity UDN in terms of Area Spectral Efficiency (ASE) for both HTC and MTC and the density of supported MTC devices. Extensive simulations were conducted to verify the accuracy of the mathematical analysis under different system parameters. Results show the existence of an optimum number of SCs to which an HTC user may connect under backhaul capacity constraints. Besides, the proposed multiple association scheme significantly improves the performance of MTC in terms of both ASE and density of supported devices. 
\end{abstract}
\section{Introduction}
According to \cite{Cisco2020}, the fastest growing mobile category between 2018 and 2023 will be Machine-to-Machine (M2M) communications. It will grow at a $19\%$ Compound Annual Growth Rate (CAGR) or nearly 2.4 folds. This reflects an increase from around 6.1 billion devices in 2018 to approximately 14.7 billion devices by 2023. Within the M2M category, connected car applications will be the fastest growing category with $30\%$ CAGR. In the same interval, smartphones will grow at a $7\%$ CAGR (or 1.4 fold) reflecting the second fastest growing category with an increase from 4.9 billion devices to 6.7 billions. Hence, it becomes clear that coexistence of HTC and MTC is inevitable in future cellular communications \cite{8712527, 8713691, 8516289}. However, a Machine-Type Communication Device (MTCD) should be handled differently compared to a Human-Type Communication User (HTCU) \cite{9022993, 8292419}. 
In Machine-Type Communication (MTC), devices will communicate with each other with minimal human intervention. MTCDs with their small packet-sizes, their massive numbers, and required massive number of simultaneous connections impose significant challenges on the next generations of cellular networks.

In an Ultra-Dense Network (UDN) environment \cite{7476821}, the cellular network can be seen as a mobile network following the users. In other words, the serving Small Cell (SC) will be always close to the user, which significantly enhances the quality of the radio link. Besides, the evolution of today's mobile broadband services necessitates the next generation of cellular networks to be capable of providing much higher end-user data rates. While UDN can significantly enhance the radio link by shortening the distances between transmitters and receivers, backhaul links may impose practical capacity limitations. In particular, with the very high density of SCs, it becomes challenging to support them by fiber links leading to limited backhaul capacities \cite{7306534}. Hence, multiple associations of SCs \cite{7931666} can mitigate such limitations in the backhaul capacity. However, the tremendous Inter-Cell Interference (ICI) found in UDN must be mitigated by adopting idle mode capabilities of the SCs \cite{8713691}. In other words, only an HTCU should be allowed to activate one or more SCs. On the other hand, to exploit the UDN environment, MTCDs should also be able to activate the nearest SC as well. However, this may lead to cases where almost all SCs are activated due to the very high density of MTCDs. In the same scope, many of the operating bands defined by the 3GPP for the New Radio (NR) in \textit{Frequency Range 1} and all of the operating bands in \textit{Frequency Range 2} are unpaired bands, \ie Time Division Duplex (TDD) is used for the same frequency band for both uplink and downlink \cite{DAHLMAN201827}. Hence, activating all SCs simultaneously will significantly deteriorate the performance of the downlink of the HTCUs.

In a different scope, fiber links are known to provide a capacity of more than 10 Gbps with very limited latency in the order of hundreds of microseconds \cite{7456186}. Hence, they are considered as the optimal backhaul choice. However, connecting all cells in a UDN with fiber links can be challenging. Deployment cost, and deployment time are very high, and in some scenarios it is not even applicable with the massive numbers of small cells. Hence, wireless backhauling can be a promising alternative which, however, suffers from limitations on the achievable capacities \cite{7306534}. In this regard, backhaul capacity constraints (or fronthaul capacity constraints in a Cloud Radio Access Network (CRAN)) have been tackled in many works in the literature \cite{9050646, 8883218, 8423651, 7809091, 8467547, 8826383}. 

For example, in \cite{9050646}, the authors investigated the effect of a limited fronthaul capacity on the downlink performance of a heterogeneous CRAN. In doing so, they considered a hybrid Millimeter-Wave (mmWave) and free space optical fronthaul links. It was shown that different fronthaul capacities associated to the Remote Radio Heads (RRHs) require different biasing factors of these RRHs to provide a better coverage. Integrated Access and Backhaul (IAB) has been considered in \cite{8514996, 8493520, 8882288}. In \cite{8514996}, the authors considered mmWave bands and investigated both throughput and communication latency in an IAB scenario. In \cite{8493520, 8882288}, the authors analyzed the downlink rate coverage probability using mmWave for IAB under different bandwidth partitioning strategies. Alternatively, the works in \cite{8690797, 8891922} analyzed the performance of a finite fronthaul capacity cell-free massive Multiple-Input Multiple-Output (MIMO) system in the downlink and uplink, respectively.  

In this paper, we tackle the backhaul capacity limitations problem using multiple association of SCs to an HTCU. By doing so, the required high data rates by HTCUs can be split among multiple SCs to match the available backhaul capacities. Besides, more SCs will be activated and higher number of MTCDs can be supported. Unlike the work in \cite{7931666}, we study the effect of the limited backhaul capacity on the system, consider a fixed bandwidth, and investigate the effect of multiple association on the MTC performance. The main contributions in this paper can be summarized as:
\begin{itemize}
    \item We propose a multiple association for HTCUs under fixed bandwidth allocation and limited backhaul capacity to improve the achievable Area Spectral Efficiency (ASE). We show that an optimum \textit{MultiCell} size exists which depends on the backhaul link capacities, density of SCs, and density of users.
    \item To efficiently share the available bandwidth among the multiple cells serving the same HTCU, we consider the different scenarios of how the HTCUs may be associated to SCs including those scenarios where a conflict may exist among the different users. 
    \item In parallel, we study the MTC uplink performance under the proposed multiple association scheme and show how it can increase both the density of supported MTCDs and the achievable ASE by those devices.
    \item Using tools from stochastic geometry, we provide  closed form expressions for the achievable ASE in HTC and MTC and the density of supported MTCDs. 
\end{itemize}

The rest of the paper is organized as follows. In Section II, we describe the system model and problem formulation. Section III provides the necessary analysis for both HTC and MTC and concludes by giving closed form expressions for the different performance metrics. The obtained  Monte-Carlo simulation results and analytical expressions are reported and discussed in Section IV. Finally, Section V concludes the findings in the paper. 

\section{System Model}
\begin{figure}
	\centering
	\includegraphics[width=0.9\textwidth]{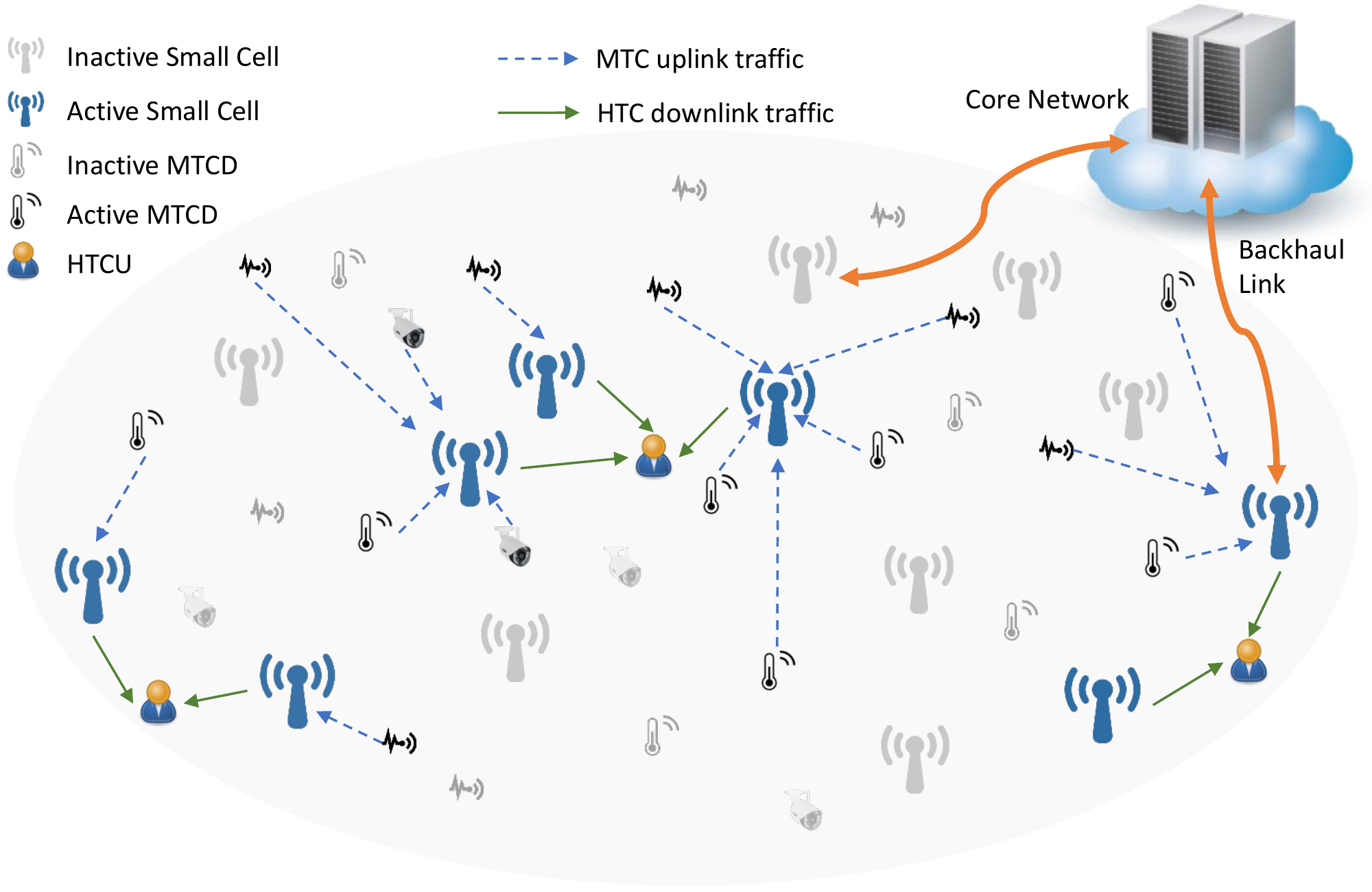}
	\caption{An illustration of the considered system model.}
	\label{fig:system_model}
\end{figure}
We consider a UDN environment with massive number of limited-backhaul SCs serving HTCUs and MTCDs. In our system model, we consider a downlink scenario of Human-Type Communication (HTC) where each HTCU may associate to more than one cell. In parallel, we study the effect of this multiple association on the performance of uplink MTC. Such a scenario lies under the use cases of both Enhanced Mobile Broad-Band eMBB and massive MTC (mMTC) targeted in the 5G and beyond \cite{9022993}. All of the SCs, HTCUs, and MTCDs are spatially distributed according to three independent Homogeneous Poisson Point Processes (HPPP)s, $\Phi_s$, $\Phi_h$, and $\Phi_m$ with intensities $\lambda_s$, $\lambda_h$, and $\lambda_m$, respectively. In a UDN environment, the density of SCs is much higher than density of HTCUs, \ie $\lambda_s \gg \lambda_h$ \cite{7476821}. However, this is not the case for the MTCDs where a heavily loaded regime is assumed such that $\lambda_m \gg \lambda_s$ which coincides with the mMTC scenario even under a UDN assumption. For a practical scenario, only a fraction of the existing MTCDs will be active at a certain time instant \cite{8903561}. Hence, we assume a fraction $\eta$ of MTCDs will be active which yields a thinned HPPP $\Phi_m^a \subset \Phi_m$ with density $\lambda_m^a=\eta \lambda_m$ for the active MTCDs. An illustration of the considered system model is shown in Fig. \ref{fig:system_model}. 

In the downlink, we consider ICI from only active cells where at least one HTCU is served by the small cell. Also, the achievable ASE comes solely from those active cells. Similarly, for the uplink MTC, the achievable ASE and the ICI both come solely from the supported active MTCDs. We assume that each SC can support a maximum number of MTCDs equal to the number of available Resource Blocks (RB)s. If the number of associated active MTCDs to a specific cell is less than the available number of RBs, then, those active MTCDs are randomly distributed over the RBs such that each MTCD exploits a single RB. Alternatively, if the number  of  associated  active  MTCDs  is larger than the available number of RBs, only a number of active MTCDs equal to the number of RBs can be supported by this cell.

All SCs, HTCUs, and MTCDs are equipped with single omnidirectional antennas. We assume a traditional path loss model in which the signal attenuates with distance $d$ as $d^{-\alpha}$ where $\alpha>2$ is the path loss exponent. For the multi-path fading, we assume Rayleigh fading channels where the channel gains are exponentially distributed with unit mean. Besides, we assume a block fading model such that the channel gain is fixed over a Transmit Time Interval (TTI) and changes independently from one TTI to another. In addition, we assume that all active transmitting nodes (SCs and MTCDs) have infinitely backlogged packets to transmit.

\subsection{Limited Backhaul Capacity}
Taking into consideration the difficulties and challenges in supporting the massive number of SCs with sufficiently large capacities in the backhaul links, we assume limited backhaul capacities for the SCs in the downlink traffic \cite{7456186}. However, for the uplink, the MTC traffic consists usually of small packets accompanied with low data rates. Hence, one does not need to consider specific limitations on the backhaul link capacities. To further illustrate, we assume that each SC is supported by a fixed limited normalized backhaul capacity $\rho$ (bps/Hz) in the radio link, \ie the supported downlink rate by each SC per one Hz is upper bounded by $\rho$. This assumption reflects a more practical scenario where the provided backhaul capacity is proportional to the allocated bandwidth in the radio link \cite{8891922}. Hence, the instantaneous achievable rate per cell per one Hz is
\begin{equation}
    \hat{R}=\min (R,\rho)
    \label{eq:R_min}
\end{equation}
where $R$ is the instantaneous achievable rate per cell per one Hz in the radio link.

\subsection{Human-Type Communication (HTC) and Multiple Association}
We refer to the conventional association scheme where each HTCU connects to and activates the nearest SC as the single association scheme. In such scheme, we expect two different scenarios illustrated in Fig. \ref{fig:single_assoc_scenarios}. In the first scenario, \textit{S1}, an active SC serves exactly one HTCU while in the second scenario, \textit{S2}, an active SC serves more than one HTCU. In the latter, we assume that the multiple users served by the same SC will share the available RBs orthogonally based on a Frequency Division Multiple Access (FDMA) approach. In doing so, the achievable data rate from a specific SC is divided among the multiple users such that no mutual interference exists among them. However, the considered SC will generate ICI to the neighboring cells over the whole bandwidth regardless of the number of served HTCUs.

\begin{figure}
	\centering
	\includegraphics[width=0.45\textwidth]{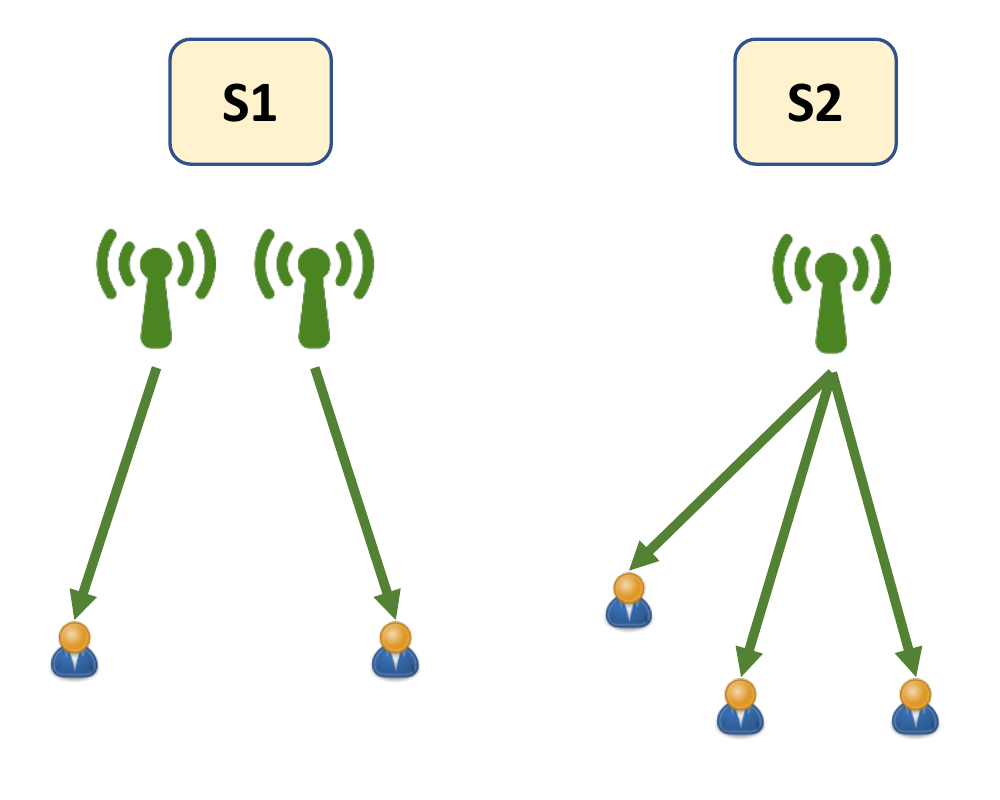}
	\caption{Different scenarios in single association scheme.}
	\label{fig:single_assoc_scenarios}
\end{figure}

For a double (multiple) association scheme, each HTCU connects to and activates the first and the second (up to $M$) nearest SCs. We refer to the group of two (or more in case of multiple association) cells as a \textit{MultiCell}. Also, the set of the first nearest SCs to the HTCUs are denoted by \textit{Tier-1}, the second nearest set of SCs by \textit{Tier-2}, and so on. Low-ordered tiers refer to the near cells while higher ordered tiers refer to further cells. Since the available spectrum is fixed and the frequency reuse factor is one, we assume that the available frequency will be divided equally among the different tiers. It is noteworthy that the more bandwidth given to \textit{Tier-1}, the higher the spectrum efficiency will be. In other words, the highest spectrum efficiency can be achieved when the whole frequency band is allocated to \textit{Tier-1} (single association), however, the backhaul capacities will be the limiting factor. 
In this paper, we investigate the system performance under equal frequency allocation versus the \textit{MultiCell} size.

Similar to single association, in a double association scheme, an active SC may serve exactly one user or more than one user as shown in Fig. \ref{fig:double_assoc_scenarios}. Different from single association, an SC serving more than one user may be associated to users on different tiers (\textit{S3}). We refer to those users as \textit{conflicting users}. This will result in severe interference on the users served on the higher-ordered tiers, \ie the interference will be higher than the useful signal. To overcome this challenge, one may interchange the frequency bands associated to the different tiers of the \textit{conflicting users}.  However, this will require optimization over the whole network which may not be practically feasible. Alternatively, one may switch-off (disconnect) the higher tier(s) of those \textit{conflicting users} as illustrated in \textit{S3} in Fig. \ref{fig:double_assoc_scenarios}. It is noteworthy that when an SC serves more than one user on the same tier, the achievable data rate is divided among the multiple users similar to single association scheme.

\begin{figure}
	\centering
	\includegraphics[width=0.8\textwidth]{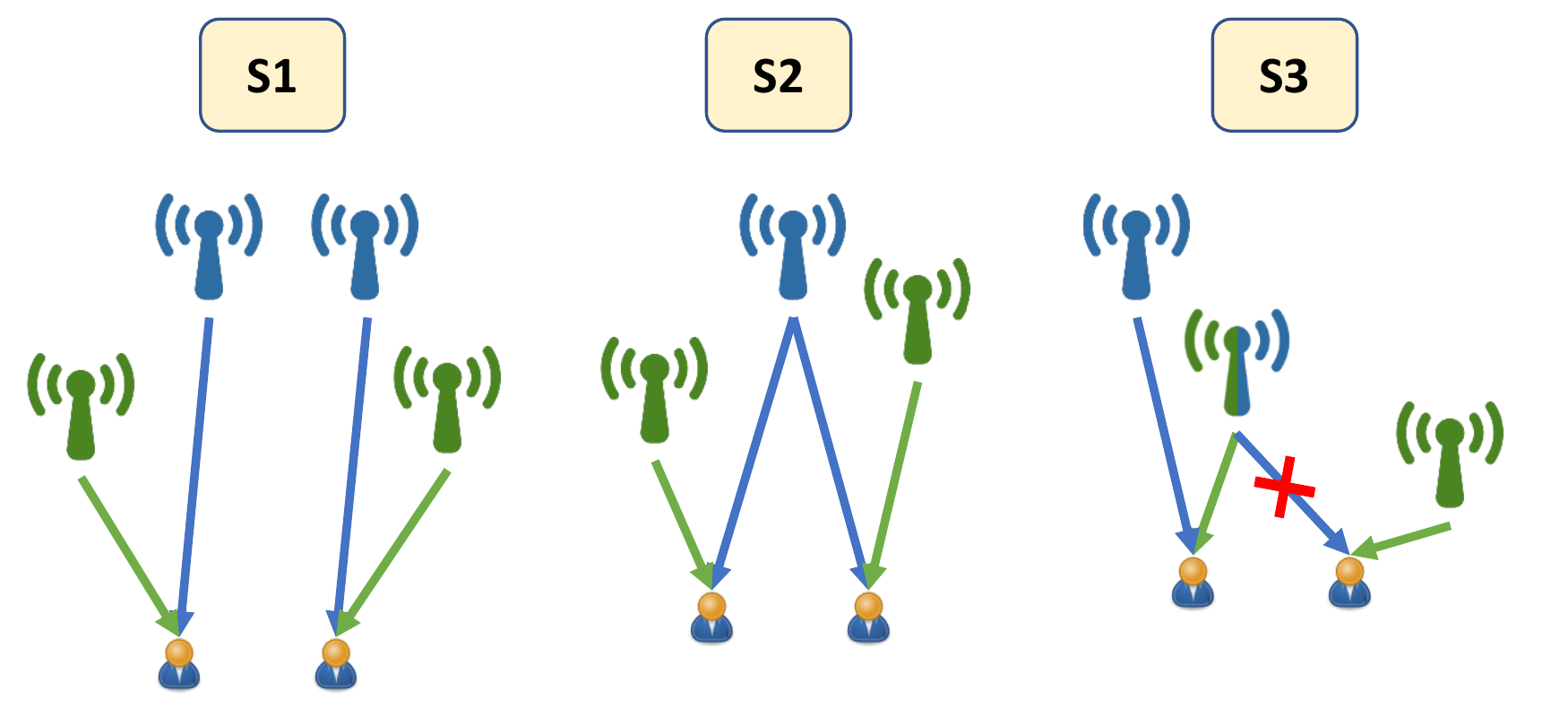}
	\caption{Different scenarios in double association scheme.}
	\label{fig:double_assoc_scenarios}
\end{figure}

In multiple association, a system with a \textit{MultiCell} size $M$ will have the active cells clustered around each HTCU. Hence, the active cells form a point process $\Phi_s^a$ thinned from the HPPP $\Phi_s$.
Fig. \ref{fig:pcp} illustrates how the active cells will be distributed around the HTCUs for different sets of small cell densities. The darkest cells are those cells belonging to \textit{Tier-1} and brighter cells represent higher-ordered tiers. In case of a cell serving more than one HTCU on different tiers, the higher-ordered tier(s) will be disconnected as explained earlier. As a consequence, we can note from Fig. \ref{fig:pcp_a} that the majority of the cells are activated as low-ordered (dark coloured) tiers when the density of SCs $\lambda_s$ is not relatively large compared to $M\lambda_h$. As $\lambda_s$ increases, more cells are activated on the different tiers while the cell areas get smaller as clear from Fig. \ref{fig:pcp_b}.

\begin{figure*}
\centering
\subfloat[$\lambda_s = 1,000~\text{cells/km}^2$]{\includegraphics[width=0.45\textwidth, height=0.45\textwidth, trim=1.5cm 1cm 2cm 1cm,clip]{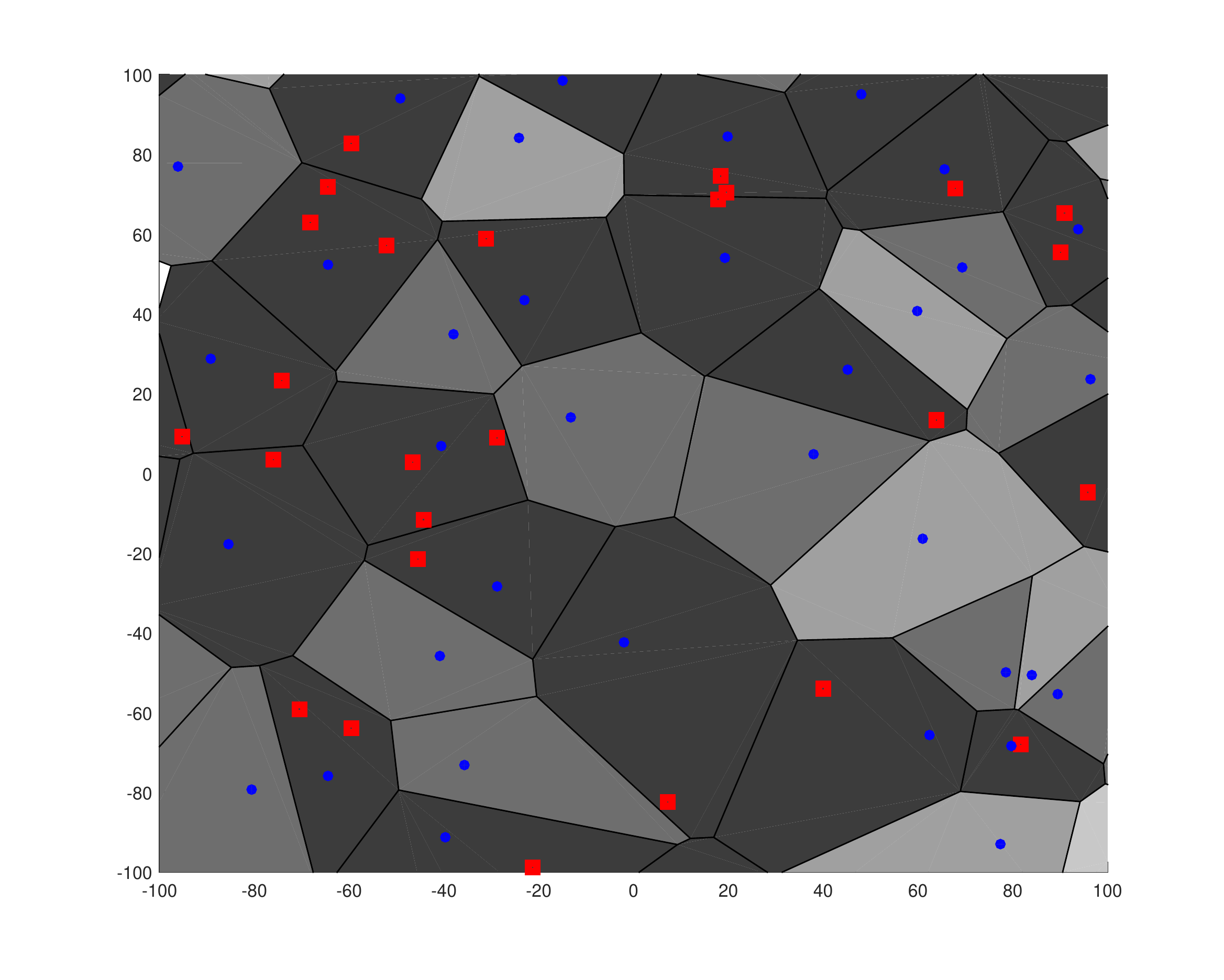}\label{fig:pcp_a}}
\subfloat[$\lambda_s = 5,000~\text{cells/km}^2$]{\includegraphics[width=0.45\textwidth, height=0.45\textwidth, trim=1.5cm 1cm 2cm 1cm,clip]{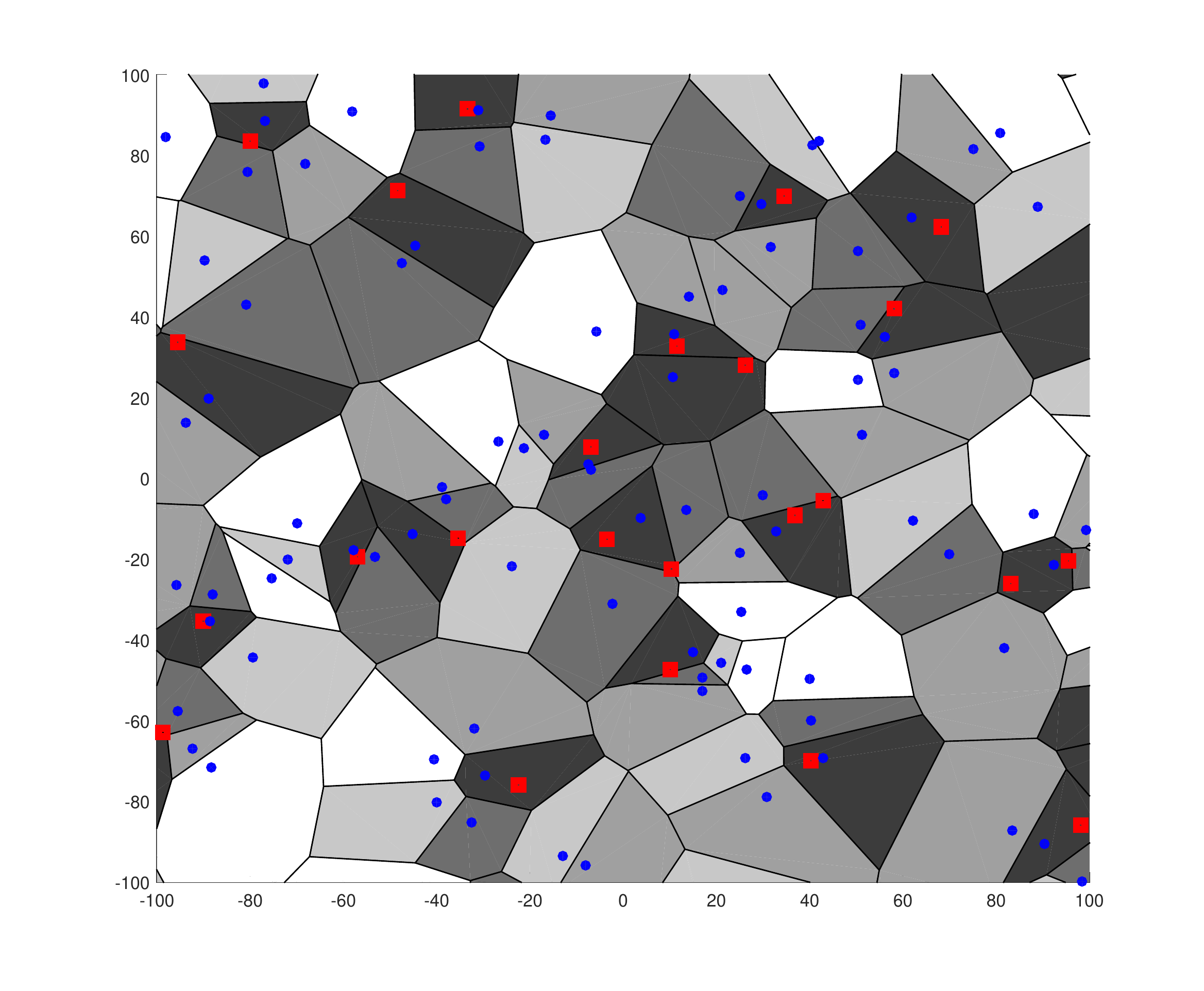}\label{fig:pcp_b}}
\caption{Different active cells realizations for different values of $\lambda_s$, where $\lambda_h=500$ (HTCUs/km$^2$) and \textit{MultiCell} size $M=5$. The red squares indicate HTCUs and the blue dots denote active cells. Dark and bright cells reflect low and high ordered tiers, respectively}
\label{fig:pcp}
\end{figure*}

\subsection{Machine-Type Communication (MTC)}
Following the multiple association scheme for the HTCUs, the active MTCDs will connect to the \textit{nearest-active} cell to upload their data. In doing so, the MTCDs will benefit from the additionally activated SCs to support their required massive number of simultaneous connections and to boost their achievable ASE. In order to investigate the achievable gain, we focus on a practical scenario where each cell has a limited allocated bandwidth. In other words, we assume a fixed number of RBs $N_{RB}$ per each  active cell where each MTCD requires exactly one RB to transmit its data to the SC. In addition, we assume that the MTCDs associated to a certain SC are randomly distributed over the available RBs in order to minimize the ICI. Hence, one may conclude that the number of supported MTCDs per cell will be upper bounded by the number of available RBs. 

It is noteworthy that when the active cells are spatially distributed according to a HPPP, the cell sizes of all cells are Independently and Identically Distributed (i.i.d.) and their distributions are given in \cite{FERENC2007}. Accordingly, from \cite{Yu2013}, the probability mass function (pmf) of the number of users/devices per cell 
is given by
\begin{equation}
    f_N(n)=\frac{(\frac{3.5}{\lambda+3.5})^{3.5}}{\Gamma(3.5)}~\Big( \frac{\lambda}{\lambda+3.5}\Big)^n~\frac{\Gamma(n+3.5)}{\Gamma(n+1)}.
    \label{eq:N}
\end{equation}

\section{Performance Analysis}
It is known that the main limiting factor for adopting single association is the limited backhaul capacity. If the achievable rate in the radio link per cell per Hz is $R$ and the backhaul capacity is $\rho~(bps/Hz)$, then, the actual achievable rate $\hat{R}$ will be as given in (\ref{eq:R_min}). Hence, we investigate the effect of the proposed multiple association scheme under such backhaul capacity limitations by providing closed-form expression for the achievable ASE for HTC. In addition, we proceed with the analysis to study the effect of the proposed multiple association scheme on the MTC performance in terms of both density of supported connections and achievable ASE.

\subsection{Multiple Association}

In the proposed multiple association scheme, the number of active cells depends on the relative densities between SCs and HTCUs as well as the \textit{MultiCell} size. Since the HTCUs are uniformly distributed according to HPPP, it is possible to have some HTCUs which are in close proximity to each other. Hence, these close-by HTCUs may be associated to common cells either on the same or on different tiers. Following our assumption illustrated in Fig. \ref{fig:double_assoc_scenarios}, any SC serving more than one user on different tiers will be disconnected on the higher-ordered tier(s). As a result, we expect that the density of active cells on each tier will decrease as the tier order increases. It is to be noted that the density of active cells on a certain tier $k$, $\lambda_k^a$ will be independent of the higher-ordered tiers. For example, the density of active cells on \textit{Tier-1} $\lambda_1^a$ is independent from the \textit{MultiCell} size $M$. Hence, using the probability mass function (pmf) of the number of users per cell given in (\ref{eq:N}), the probability of cell activation on \textit{Tier-1} can be calculated as
\begin{equation}
    p_1^a=1-f_N(0)=1-\left( \frac{3.5}{3.5+\frac{\lambda_h}{\lambda_s}}\right)^{3.5}.
    \label{eq:p_a_1}
\end{equation}
For higher order tiers, \textit{Tier-k} where $1<k\le M$, Lemma \ref{lem:lambda_k} gives the density of active cells in each tier.

\begin{lem}
The density of active cells on a specific Tier-k, $1<k\le M$ is given by
\begin{equation}
    \lambda_k^a=\mathcal{C} \left( \frac{3.5}{3.5+\frac{k\lambda_h}{\lambda_s}}\right)^{3.5} \lambda_s,
    \label{eq:lambda_k}
\end{equation}
\begin{equation}
    \mathcal{C}=\left( \frac{3.5 \lambda_s+k \lambda_h}{3.5 \lambda_s+(k-1) \lambda_h} \right)^{3.5} -1  .
    \label{eq:lambda_k_correction}
\end{equation}
\label{lem:lambda_k}
\end{lem}
\begin{proof}
Let $\lambda_k^{all}$ represent the density of all active cells on all tiers up to \textit{Tier-k}. The HTCUs are spatially distributed according to an HPPP such that each HTCU associates to the nearest $k$ cells ($k\le M$). Alternatively, we follow the assumption in \cite{7511520} and replace each HTCU by $k$ users which are uniformly distributed such that each user associates only to the nearest cell. Hence, the density of HTCUs becomes $k\lambda_h$ and the activation probability of all active cells up to \textit{Tier-k} is given by
\begin{equation}
    \lambda_k^{all}~=~p_k^{all}~\lambda_s~=~\left[1-\Big( \frac{3.5}{3.5+\frac{k\lambda_h}{\lambda_s}}\Big)^{3.5}\right] \lambda_s.
    \label{eq:lambda_k_all}
\end{equation}
Next, we evaluate the density of active cells on \textit{Tier-k} by subtracting $\lambda_k^{all}-\lambda_{k-1}^{all}$ which gives (\ref{eq:lambda_k}) and completes the proof. 
\end{proof}

As the tier order increases in (\ref{eq:lambda_k_all}), the activation probability of all cells up to \textit{Tier-k} increases. However, from (\ref{eq:lambda_k}) and (\ref{eq:lambda_k_correction}) one can conclude that for higher $k$, the density of active cells on \textit{Tier-k} gets smaller. 

\subsection{Human-Type Communication}
To investigate the performance of HTC, we assume a typical HTCU located at the origin which connects to and activates the nearest $M$ cells. This typical user reflects the performance of all users existing in the network according to Slivnyak's theorem \cite{haenggi_2012}. The achievable rate by this typical user associated with a SC on \textit{Tier-k}, reflects the average achievable rate per cell per one Hz of that tier. From the HPPP distribution of the SCs, the probability that a circle of radius $r$ centered at the typical user includes exactly $k$ SCs is
\begin{equation}
\mathbb{P}(n=k;r) = \frac{{(\pi \lambda_s r^2)}^k}{k!} e^{-\pi \lambda_s r^2}.
\end{equation}
Let $r_k$ be the distance between the typical user and the $k^{th}$ nearest SC. 
Hence, the probability density function (pdf) of the distance between the typical user and the $k^{th}$ nearest SC is given by \cite{haenggi_2012}
\begin{equation}
    f_{r_k}(r)=\frac{2(\pi \lambda_s)^k}{(k-1)!} r^{2k-1} e^{-\pi \lambda_s r^2}.
    \label{eq:r_k}
\end{equation}
The rate achievable by the typical user in the radio link from the $k^{th}$ tier is given by
\begin{equation}
    R_k=\log_2(1+\gamma_k),
    \label{eq:R_k}
\end{equation}
In (\ref{eq:R_k}), $\gamma_k$ is the instantaneous Signal-to-Noise plus Interference Ratio (SINR) for the typical user, given by,
\begin{equation}
    \gamma_k=\frac{P h_k r_k^{-\alpha}}{I_k+\sigma^2}=\frac{h_k r_k^{-\alpha}}{I'_k+\frac{\sigma^2}{P}},
    \label{eq:gamma_k}
\end{equation}
with $P$ is the transmission power of the SC, $h_k$ and $r_k$ are the channel gain and the distance between the typical user and its serving cell on the $k^{th}$ tier, respectively, and $\sigma^2$ is the power of the Additive White Gaussian Noise (AWGN). Furthermore, the term $I'_k$ represents the normalized ICI at the typical user received from tier $k$. Since the bandwidth is shared orthogonally among the different tiers, the normalized ICI is given by 
\begin{equation}
    I'_k=\sum_{j\in \Phi_k^a \BSlash b_{k_0}}h_j r_j^{-\alpha},
    \label{eq:I_k}
\end{equation}
where $\Phi_k^a$ is the set of all active cells on \textit{Tier-k}, $b_{k_0}$ is the serving cell on \textit{Tier-k}, and $h_j$ and $r_j$ are the channel gain and distance between the typical user and the interfering cell belonging to \textit{Tier-k}, respectively. 

\begin{thm}
In a multiple association scheme with \textit{MultiCell} size $M$, the average achievable rate per cell on Tier-k, assuming infinite backhaul capacity, is given by
\begin{equation}
    \Bar{R_k^M}=\int_0^{\infty}\frac{1}{\left(1+p_k^a(2^{Mt}-1)^{\frac{2}{\alpha}}\int_{(2^{Mt}-1)^{-\frac{2}{\alpha}}}^{\infty}\frac{1}{1+u^{\frac{\alpha}{2}}}\mathrm{d}u\right)^k}\mathrm{d}t,
\end{equation}
where $p_k^a$ is obtained from (\ref{eq:lambda_k}) with $p_k^a=\frac{\lambda_k^a}{ \lambda_s}$.
\label{thm:R_k_M}
\end{thm} 
\begin{proof}
Given the equal bandwidth allocated to each tier, the instantaneous achievable rate by the typical user from its associated cell on \textit{Tier-k} is given by
\begin{equation}
    R_k^M=\frac{1}{M}\log_2(1+\gamma_k).
\end{equation}
With $\gamma_k$ given in (\ref{eq:gamma_k}), one can evaluate the complementary cumulative distribution function (CDF) of $\gamma_k$ as follows,
\begin{align}
    \mathbb{P}[\gamma_k>\zeta]
    &=~\mathbb{E}_{I'_k,r_k}\left[\mathbb{P}\Big[h_k>\zeta(\frac{\sigma^2}{P}+I'_k)r_k^{\alpha}\Big] ~ | ~ I'_k,r_k \right]\notag\\
    &\overset{(a)}{=}\mathbb{E}_{I'_k,r_k}\left[\exp\Big( -\zeta(\frac{\sigma^2}{P}+I'_k)r_k^{\alpha} \Big) ~ | ~ I'_k,r_k \right]\notag\\
    &=\int_{r>0}e^{-\zeta \frac{\sigma^2}{P}r^{\alpha}}\mathbb{E}_{I'_k}\left[\exp\Big( -\zeta I'_k r^{\alpha} \Big) ~ | ~ I'_k,r \right]f_{r_k}(r)~\mathrm{d}r \notag\\
    &=\int_{0}^{\infty}e^{-\zeta \frac{\sigma^2}{P}r^{\alpha}}\mathcal{L}_{I'_k}(\zeta r^{\alpha};r)f_{r_k}(r)~\mathrm{d}r,
    \label{eq:SIR_h_k_ccdf}
\end{align}
where $\mathbb{E}_X[.]$ denotes expectation over the random variable $X$ and $(a)$ follows from the unit mean exponential distribution of the channel gain $h_k$ corresponding to the Rayleigh small scale fading. $\mathcal{L}_{I'_k}(s)=\mathbb{E}_{I'_k}[e^{-sI'_k}]$ is the Laplace transform of the normalized ICI on \textit{Tier-k}, $I'_k$, given in (\ref{eq:I_k}). Hence,
\begin{align}
    \mathcal{L}_{I'_k}(s;r)&=\mathbb{E}_{I'_k}\left[e^{-sI'_k}\right] \notag \\
&=\mathbb{E}_{\Phi_k^a, h}\left[\exp\left({-s \sum\limits_{j\in \Phi_k^a \BSlash b_{k_0}}h_j r_j^{-\alpha}}\right)\right] \notag \\
    &=\mathbb{E}_{\Phi_k^a}\left[\prod_{j\in \Phi_k^a \BSlash b_{k_0}} \mathbb{E}_{h_j}\Big[\exp(-s h_j r_j^{-\alpha})\Big]\right].
    \label{eq:lap_h_1st_part}
\end{align}
From the probability generating functional (PGFL) of a HPPP, one can find that \cite{6042301}
\begin{equation}
    \mathbb{E}\left[\prod_{z\in\Phi}f(z)\right]=\exp\left(-\lambda'\iint\limits_{\mathbb{R}^2}(1-f(z)) \mathrm{d}z\right),
    \label{eq:pgfl}
\end{equation}
where $\lambda'$ is the density of the HPPP $\Phi$ and $\mathbb{R}^2$ is the two dimensional Euclidean space. Since the power of the channel gain, $h$, is exponentially distributed with mean $\mu$, then,
\begin{equation}
    \mathbb{E}_h\left[e^{-\varrho h} \right]=\frac{1}{1+\mu \varrho}.
    \label{eq:exp_h}
\end{equation}
Using (\ref{eq:lap_h_1st_part}), (\ref{eq:pgfl}), and (\ref{eq:exp_h}),
\begin{align}
    \mathcal{L}_{I'_k}(s;r)&=\exp\left(-\lambda_k^a\int\limits_0^{2\pi}\int\limits_r^{\infty} \left(1-\mathbb{E}_{h_j}\left[e^{-s {h_j} {r_j}^{-\alpha}}\right] \right){r_j}~\mathrm{d}{r_j}~\mathrm{d}\theta\right) \notag \\
&=\exp\left(-2\pi\lambda_k^a\int\limits_r^{\infty} \left(1-\frac{1}{1+s{r_j}^{-\alpha}} \right){r_j}~\mathrm{d}{r_j}\right) \notag \\
    &=\exp\left(-\pi~\lambda_k^a~s^{\frac{2}{\alpha}} \int_{\frac{r^2}{s^{2/\alpha}}}^{\infty}\frac{1}{1+u^{\frac{\alpha}{2}}}\mathrm{d}u\right),
    \label{eq:lap_h}
\end{align}
where the last step is obtained by setting $u={r_j}^2s^{-2/\alpha}$.

Substituting $\mathcal{L}_{I'_k}(\zeta r^{\alpha};r)$ from (\ref{eq:lap_h}) into (\ref{eq:SIR_h_k_ccdf}), and assuming an interference-limited scenario ($\sigma^2=0$), the complementary CDF of the Signal-to-Interference Ratio (SIR) can be given as
\begin{align}
&\mathbb{P}[\gamma_k>\zeta]\notag\\
    &=\int\limits_{0}^{\infty} e^{-\pi\lambda_k^a r^2\zeta^{\frac{2}{\alpha}} \int_{\zeta^{\frac{-2}{\alpha}}}^{\infty}\frac{1}{1+u^{\frac{\alpha}{2}}}\mathrm{d}u}~\frac{2(\pi \lambda_s)^k}{(k-1)!} r^{2k-1}e^{-\pi \lambda_s r^2} ~\mathrm{d}r \notag\\ 
&\overset{(a)}{=}\int\limits_{0}^{\infty} e^{-\pi\lambda_s \mathcal{A} r^2 }~\frac{2(\pi \lambda_s)^k}{(k-1)!} r^{2k-1} ~\mathrm{d}r \notag\\
    &=\int\limits_{0}^{\infty}  \frac{(\pi \lambda_s)^{(k-1)}}{(k-1)!} r^{2(k-1)} \boldsymbol{\cdot} e^{-\pi\lambda_s \mathcal{A} r^2 }~2\pi\lambda_s r ~\mathrm{d}r \notag\\
&\overset{(b)}{=}\frac{1}{{\left(1+p_k^a\zeta^{\frac{2}{\alpha}} \int_{\zeta^{\frac{-2}{\alpha}}}^{\infty}\frac{1}{(1+u^{\frac{\alpha}{2}})}\mathrm{d}u\right)}^k},
\label{eq:gamma_k_h_ccdf}
\end{align}
where (a) follows from setting $\mathcal{A}=1+p_k^a\zeta^{\frac{2}{\alpha}} \int_{\zeta^{\frac{-2}{\alpha}}}^{\infty}\frac{1}{1+u^{\frac{\alpha}{2}}}\mathrm{d}u$, and (b) is obtained using integration by reduction such that $\mathcal{I}_k=\mathbb{P}[\gamma_k>\zeta]=\frac{\mathcal{I}_{k-1}}{\mathcal{A}}$ and $\mathcal{I}_1=1 /\mathcal{A}$. Finally, the average achievable is obtained as
\begin{align}
    \Bar{R_k^M}&=\mathbb{E}\left[ R_k^M \right]=\int_0^{\infty}\mathbb{P}\left[R_k^M>t \right] \mathrm{d}t\notag\\
    &=\int_0^{\infty}\mathbb{P}\left[\gamma_k>2^{Mt}-1 \right] \mathrm{d}t.
\end{align}
Using (\ref{eq:gamma_k_h_ccdf}) completes the proof.
\end{proof}

\subsection{Limited-Backhaul Capacity}
In the previous subsection, we assumed infinite backhaul capacity between the small cells and the core network. However, to account for the major challenge of limited backhaul capacity in UDNs, we evaluate herein the instantaneous achievable rate under such a scenario given in (\ref{eq:R_min}).

\begin{cor}
In a multiple association scheme with \textit{MultiCell} size $M$ and normalized backhaul capacity $\rho$, the average achievable rate per cell on Tier-k is given by
\begin{equation}
    \hat{R}_k^M=\int_0^{\rho}\frac{1}{\left(1+p_k^a(2^{Mt}-1)^{\frac{2}{\alpha}}\int_{(2^{Mt}-1)^{-\frac{2}{\alpha}}}^{\infty}\frac{1}{1+u^{\frac{\alpha}{2}}}\mathrm{d}u\right)^k}\mathrm{d}t
    \label{eq:R_k_min}
\end{equation}
where $p_k^a$ is obtained from (\ref{eq:lambda_k}) where $\lambda_k^a=p_k^a \lambda_s$.\end{cor}
\begin{proof}
Under a normalized backhaul capacity constraint $\rho$, the average achievable rate per cell on Tier-k is given by,
\begin{align}
    \hat{R}_k^M&=\mathbb{E}\left[\min \left(\rho, R_k^M\right) \right]\notag\\
    &=\int\limits_0^{\infty} \mathbb{P}\left[\min\left(\rho,\frac{1}{M}\log_2(1+\gamma_k)\right)>t \right]\mathrm{d}t\notag\\
    &=\int\limits_0^{\infty} \mathbb{P}\left[\rho>t, \frac{1}{M}\log_2(1+\gamma_k)>t \right]\mathrm{d}t\notag\\
    &=\int\limits_0^{\rho}\mathbb{P}\left[ \gamma_k>2^{Mt}-1\right]\mathrm{d}t
\end{align}
using (\ref{eq:gamma_k_h_ccdf}) completes the proof.
\end{proof}

\begin{cor}
In  a  multiple  association  scheme  with  \textit{Multi-Cell} size $M$ and normalized backhaul capacity $\rho$, the average achievable ASE of the HTCUs, $\mathcal{T}_h$ is given by
\begin{equation}
\mathcal{T}_h=\sum_{k=1}^{M} \lambda_k^a ~\hat{R}_k^M
\label{eq:ase_h}
\end{equation}
where $\lambda_k^a$ is the density of active cells on \textit{Tier-k} given in (\ref{eq:lambda_k}), and $\hat{R}_k^M$ is the average achievable rate per cell on \textit{Tier-k} under limited backhaul constraints given in (\ref{eq:R_k_min}).
\end{cor}
\begin{proof}
The proof follows directly from the definition of ASE which is the total network achievable rate per one Hz within a unit area of one km$^2$. 
\end{proof}

\subsection{Machine-Type Communication}

The active MTCDs form a thinned HPPP $\Phi_m^a$ such that $\Phi_m^a\subseteq\Phi_m$, and $\lambda_m^a=\eta\lambda_m$. These active devices are uniformly distributed over the area under study where each active MTCD connects to the nearest active cell. The active cells, however, form another thinned point process from $\Phi_s$ with density $\lambda_M^{all}$ given in (\ref{eq:lambda_k_all}). In this regard, we assume that each active cell has a limited number of RBs and that each MTCD requires exactly one RB to transmit its data. Hence, the number of supported devices by each active cell is upper bounded by the number of available RBs ($N_{RB}$) such that
\begin{equation}
    N_{m,s}^{cell}=\min(N_m^{cell},N_{RB}),
    \label{eq:N_m_s_cell}
\end{equation}
where $N_m^{cell}$ is the number of associated MTCDs to a specific cell and $N_{m,s}^{cell}$ is the corresponding number of supported devices by this cell.

The density of supported MTCDs represents a key performance metric for MTC. Although MTCDs tend to upload small packets of data, it is challenging to support their required massive number of connections. Hence, increasing the density of supported MTCDs constitutes a major challenge which can be tackled by adopting the proposed multiple association scheme. In doing so, as we increase the \textit{MultiCell} size $M$, the density of active cells $\lambda_M^{all}$ given in (\ref{eq:lambda_k_all}) increases. Hence, the ability of the network to simultaneously support a higher number of MTCDs significantly improves. 
It is to be noted that the activated cells on different tiers vary in both of their densities and cell size distributions as illustrated earlier in Fig. \ref{fig:pcp}. Accordingly, the density of associated MTCDs on each tier of active cells will vary significantly. However, it is more of concern to find the average number of MTCDs per cell, which may also differ from one tier to another depending on the cell size distribution of each tier. However, for the considered densities of SCs and HTCUs where, one can notice from Fig. \ref{fig:pcp} that the different cells on different tiers seem to have identical size distribution.
Hence, for tractability, when investigating the performance of MTC, we assume that the active cells are uniformly distributed according to an HPPP $\Phi_s^a$ with density $\lambda_M^{all}$ evaluated from (\ref{eq:lambda_k_all}). Following this assumption, the probability mass function (pmf) of the distribution of the number of MTCDs per active cell $N_m^{cell}$ can be approximated by (\ref{eq:N}) with $\lambda=\frac{\lambda_m^a}{\lambda_M^{all}}$.

\begin{lem}
The average density of supported MTCDs in a multiple association scheme with \textit{MultiCell} size $M$, number of RBs $N_{RB}$, MTCDs activation probability $\rho$, and densities $\lambda_s$, $\lambda_h$, and $\lambda_m$ for the SCs, HTCUs, and MTCDs, respectively, can be approximated by
\begin{equation}
    \lambda_m^s\approx \mathcal{J}\lambda_M^{all},
    \label{eq:lambda_m_s}
\end{equation}
\begin{equation}
    \mathcal{J}=\left(N_{RB}-\sum_{i=0}^{N_{RB}-1}\sum_{n=0}^{i}f_{N_m^{cell}}(n)\right),
    \label{eq:J}
\end{equation}
\begin{equation}
    f_{N_m^{cell}}(n)=\frac{(\frac{3.5 \lambda_M^{all}}{\rho\lambda_m+3.5 \lambda_M^{all}})^{3.5}}{\Gamma(3.5)}~\Big( \frac{\rho\lambda_m}{\rho\lambda_m+3.5 \lambda_M^{all}}\Big)^n~\frac{\Gamma(n+3.5)}{\Gamma(n+1)}.
    \label{eq:f_N_m_cell}
\end{equation}
where $\lambda_M^{all}$ is given by (\ref{eq:lambda_k_all}).
\begin{proof}
Multiplying the density of supported devices per cell $\mathbb{E}[N_{m,s}^{cell}]$ by the density of all active cells $\lambda_M^{all}$ gives (\ref{eq:lambda_m_s}). Starting from (\ref{eq:N_m_s_cell}), the density of supported devices per cell is calculated as follows,
\begin{align}
    \mathcal{J}=\mathbb{E}[N_{m,s}^{cell}]
    &=\sum_{i=0}^{\infty}\left(1-F_{N_{m,s}^{cell}}(i)\right)\notag\\
&=\sum_{i=0}^{\infty}\mathbb{P}\left[N_{m,s}^{cell}>i\right]\notag\\
    &=\sum_{i=0}^{\infty}\mathbb{P}\left[\min\left(N_m^{cell},N_{RB}\right)>i\right]\notag\\
    &=\sum_{i=0}^{\infty}\mathbb{P}\left[N_m^{cell}>i,i<N_{RB}\right]\notag\\
    &=\sum_{i=0}^{N_{RB}-1}\mathbb{P}\left[N_m^{cell}>i\right]\notag\\
    &=\sum_{i=0}^{N_{RB}-1}\left(1-\sum_{n=0}^if_{N_m^{cell}}(n)\right),
\end{align}
where $f_{N_m^{cell}}(n)=\mathbb{P}[N_m^{cell}=n]$ is obtained from (\ref{eq:N_m_s_cell}) by substituting $\lambda=\frac{\rho\lambda_m}{\lambda_M^{all}}$ leading to (\ref{eq:f_N_m_cell}) which completes the proof.
\end{proof}
\label{lem:lambda_m_s}
\end{lem}

Using stochastic geometry with the aid of Slivnyak's theorem \cite{Haenggi2012}, and without loss of generality, we assume a typical MTCD located at the origin to reflect the performance of all existing active MTCDs. This typical device may or may not be supported by the network with certain probabilities. 
Hence, when studying the MTC performance, we only consider the supported MTCDs whose density is given in Lemma \ref{lem:lambda_m_s}. In this regard, we assume a typical MTCD located at the origin which belongs to the thinned HPPP $\Phi_m^s\subseteq\Phi_m^a$ with density $\lambda_m^s$ given in (\ref{eq:lambda_m_s}). The achievable rate by this MTCD on a specific RB is given by
\begin{equation}
    R_m=\frac{1}{N_{RB}}\log_2(1+\gamma_m),
\end{equation}
\begin{equation}
    \gamma_m=\frac{P_m^0 h_m r_m^{-\alpha}}{I_m+\sigma^2}=\frac{h_m r_m^{-\alpha}}{I'_m+\frac{\sigma^2}{P_m^0}},
\end{equation}
where $\gamma_m$ is the SIR, $P_m^0$ is the transmission power of the typical MTCD which is assumed fixed, and $h_m$ and $r_m$ are the channel gain and distance between the typical MTCD and the tagged nearest active cell, respectively. $I'_m$ is the normalized ICI at the tagged SC from all other supported MTCDs from neighboring cells transmitting their data on the same frequency band of the tagged RB. 

Since the typical MTCD is being served on a single RB, and all RBs are orthogonal in frequency, the ICI is generated from only those MTCDs allocated the same RB. In order to minimize the ICI, we assume that the supported MTCDs are randomly assigned to the available RBs. Hence, when evaluating the ICI, we consider another thinned HPPP $\Phi_{m,N}^s\subseteq \Phi_m^s$ which corresponds to the supported MTCDs on a specific RB such that $\lambda_{m,N}^s=\frac{\lambda_m^s}{N_{RB}}$. Thus, the normalized ICI can be expressed as
\begin{equation}
    I'_m=\sum_{j\in\Phi_{m,N}^s\BSlash d_0}H_j R_j^{-\alpha},
\end{equation}
where $H_j,R_j$ are the channel gain and distance between the tagged cell of the typical MTCD and the interfering MTCDs belonging to the neighboring cells, respectively.

\begin{thm}
In a multiple association scheme with MultiCell size $M$ and a number of available RBs per cell $N_{RB}$, the average achievable rate per MTCD is given by
\begin{equation}
    \Bar{R}_m=\frac{1}{N_{RB}}\int\limits_0^{\infty}\frac{1}{1+\frac{\mathcal{J}}{N_{RB}}(2^{t}-1)^{\frac{2}{\alpha}}\int_{(2^{t}-1)^{-\frac{2}{\alpha}}}^{\infty}\frac{1}{1+u^{\frac{\alpha}{2}}}\mathrm{d}u}\mathrm{d}t,
    \label{eq:R_m}
\end{equation}
where $\mathcal{J}$ is given in (\ref{eq:J}).
\end{thm}
\begin{proof}
Following the same steps in the proof of Theorem \ref{thm:R_k_M}, then,
\begin{equation}
    \mathbb{P}[\gamma_m>\zeta]=\int\limits_{0}^{\infty}e^{-\zeta \frac{\sigma^2}{P_m^0}r^{\alpha}}\mathcal{L}_{I'_m}(\zeta r^{\alpha};r')f_{r_m}(r)~\mathrm{d}r,
    \label{eq:SIR_m_ccdf}
\end{equation}
where $r_m$ is the distance between the typical MTCD and its tagged cell, and $r'$ is the minimum distance from the tagged cell to the nearest interferer. Since the Voronoi cells do not have uniform shapes, $r'$ can be larger or smaller than $r_m$. However, taking into consideration that all devices are associated to their nearest cell, one can assume $r'=r_m$ \cite{7972929}. For a typical MTCD, the distribution of $r_m$ (the distance to the nearest active cell), following the same assumption of the active cells being uniformly distributed, can be expressed as
\begin{equation}
    f_{r_m}(r)=2\pi \lambda_M^{all} r ~e^{-\pi \lambda_M^{all} r^2}.
    \label{eq:r_m}
\end{equation}

Given that the supported MTCDs follow a HPPP, the ICI at the tagged SC of the typical MTCD can be approximated by the one at the origin \cite{7972929}. Following the same steps in (\ref{eq:lap_h_1st_part}) and (\ref{eq:lap_h}), the Laplace transform of the ICI ($\mathcal{L}_{I'_m}(s;r)$) can be expressed as,
\begin{align}
    \mathcal{L}_{I'_m}(s;r)&=\mathbb{E}_{\Phi_{m,N}^s, H}\left[\exp\left({-s \sum\limits_{j\in \Phi_k^a \BSlash b_{k_0}}H_j R_j^{-\alpha}}\right)\right] \notag \\
    &=\exp\left(-\pi~\lambda_{m,N}^s~s^{\frac{2}{\alpha}} \int_{\frac{r^2}{s^{2/\alpha}}}^{\infty}\frac{1}{1+u^{\frac{\alpha}{2}}}\mathrm{d}u\right).
    \label{eq:lap_m}
\end{align}

Assuming an interference-limited scenario due to the massive number of supported MTCDs, and substituting from (\ref{eq:r_m}) and (\ref{eq:lap_m}) with $s=\zeta r^{\alpha}$ into (\ref{eq:SIR_m_ccdf}), then, 
\begin{align}
\mathbb{P}&[\gamma_k>\zeta]\notag\\
    &=\int\limits_{0}^{\infty} e^{-\pi\lambda_{m,N}^s \zeta^{\frac{2}{\alpha}} \int_{\zeta^{\frac{-2}{\alpha}}}^{\infty}\frac{1}{1+u^{\frac{\alpha}{2}}}\mathrm{d}u~r^2}~(2\pi \lambda_M^{all}~r) ~e^{-\pi \lambda_s r^2} ~\mathrm{d}r \notag\\ 
&\overset{(a)}{=}\int\limits_{0}^{\infty} e^{-\pi(1+\mathcal{B})\lambda_M^{all} ~ r^2 }~2\pi\lambda_M^{all}~r ~\mathrm{d}r =\frac{1}{1+\mathcal{B}}
\label{eq:gamma_m_ccdf}
\end{align}
where $\mathcal{B}=\frac{\mathcal{J}}{N_{RB}}\zeta^{\frac{2}{\alpha}} \int_{\zeta^{\frac{-2}{\alpha}}}^{\infty}\frac{1}{(1+u^{\frac{\alpha}{2}}}\mathrm{d}u$ and ($a$) is obtained using $\lambda_{m,N}^s=\frac{\mathcal{J}\lambda_M^{all}}{N_{RB}}$ from (\ref{eq:lambda_m_s}). Finally, the average achievable rate per MTCD can be calculated from $\Bar{R}_m=\frac{1}{N_{RB}}\int_0^{\infty}\mathbb{P}[\gamma_k>2^t-1]\mathrm{d}t$ which completes the proof.
\end{proof}

\begin{cor}
In  a  multiple  association  scheme  with  \textit{Multi-Cell} size $M$, the average achievable ASE of the MTCDs, $\mathcal{T}_m$, is given by
\begin{equation}
\mathcal{T}_m=\lambda_m^s ~\Bar{R}_m
\end{equation}
where $\lambda_m^s$ is the density of supported MTCDs given in (\ref{eq:lambda_m_s}), and $\Bar{R}_m$ is the average achievable rate per MTCD given in (\ref{eq:R_m}).
\end{cor}
\begin{proof}
The proof follows directly from the definition of ASE which is the total network achievable rate per one Hz within a unit area of one $km^2$. 
\end{proof}

\section{Simulation Results}
We report the impacts of different system parameters on the performance of both HTC and MTC. In doing so, we simulate different scenarios by generating different realizations of the point processes and averaging the performance over these realizations. We consider a simulation area of $1$ km$^2$ with $500$ spatial realizations of the point processes and $10$ different realization of the small scale fading in each spatial realization. Unless otherwise stated, we consider the following parameter sets. \textit{MultiCell} size is set to $M=5$, the density of small cells $\lambda_s=5000$ cells/km$^2$, HTCUs $\lambda_h=500$ HTCUs/km$^2$, and MTCDs $\lambda_m=10^6$ devices/km$^2$. We set the activation ratio of MTCDs as $\eta=0.1$ and the path loss exponent $\alpha=4$ to reflect an urban environment. Each SC is allocated $N_{RB}=10$ resource blocks (RBs) in the uplink where each MTCD requires exactly one RB. The total bandwidth available in both downlink and uplink as well as the transmission powers are normalized to unity and the noise power is set to $-174$ dBm/Hz.

\subsection{Simulation Setup}
We first generate the spatial realization of the three different HPPPs $\Phi_s$, $\Phi_h$, and $\Phi_m$ with their corresponding densities. Then, each HTCU associates to and activates the nearest $M$ SCs. All the remaining SCs are switched to idle mode. If one cell is found to be serving more than one HTCU on different tiers, the higher tier is disconnected. Next, a fraction $\eta$ of the MTCDs is randomly selected to become active and ready to transmit in correspondence to a real scenario where not all MTCDs are simultaneously active. Then, each of those active MTCDs associates to the nearest active cell. If a cell has a number of associated MTCDs greater than the available RBs $N_{RB}$, a number of active MTCDs equal to $N_{RB}$ are randomly selected to be supported. Alternatively, if the number of associated MTCDs is less than $N_{RB}$, then, they are randomly allocated RBs to transmit their data where each MTCD is allocated exactly one RB. When the achievable traffic of the downlink HTC on the radio link exceeds the available backhaul capacity, the actual rate is then upper-bounded by the backhaul capacity. Finally, the ASE is calculated by summing all the achievable actual rates over the whole network in an area of one km$^2$ with a normalized bandwidth of one Hz.

\subsection{Human-Type Communication}
\begin{figure*}
	\centering
	\subfloat[Different $\rho$ and $\lambda_s=5,000$ SCs/km$^2$.]{\includegraphics[width=0.5\textwidth, height=0.4\textwidth]{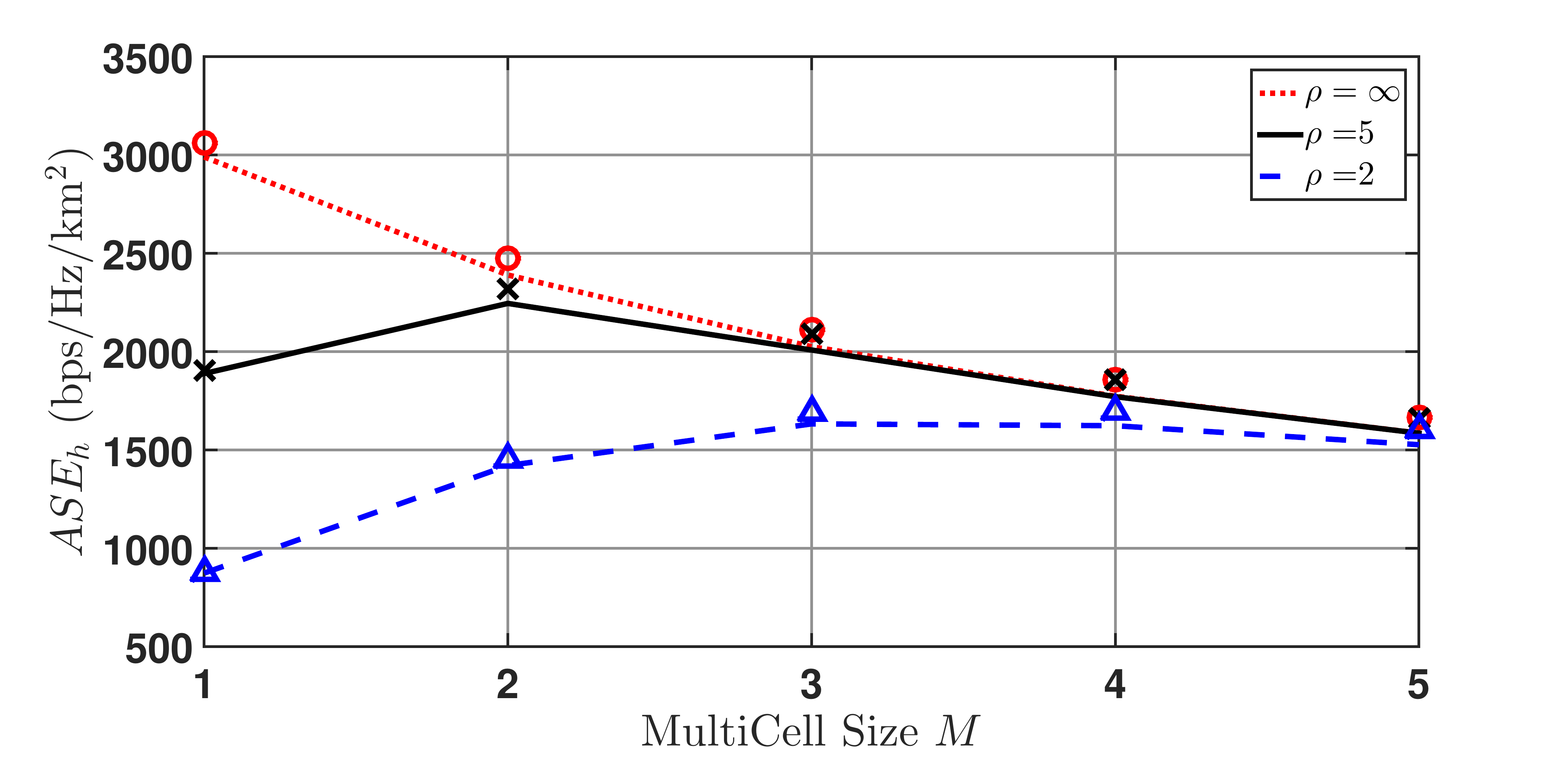}\label{fig:ASE_h_vs_M_a}}
    \subfloat[Different $\lambda_s$ and $\rho=2$ bps/Hz]{\includegraphics[width=0.5\textwidth, height=0.4\textwidth]{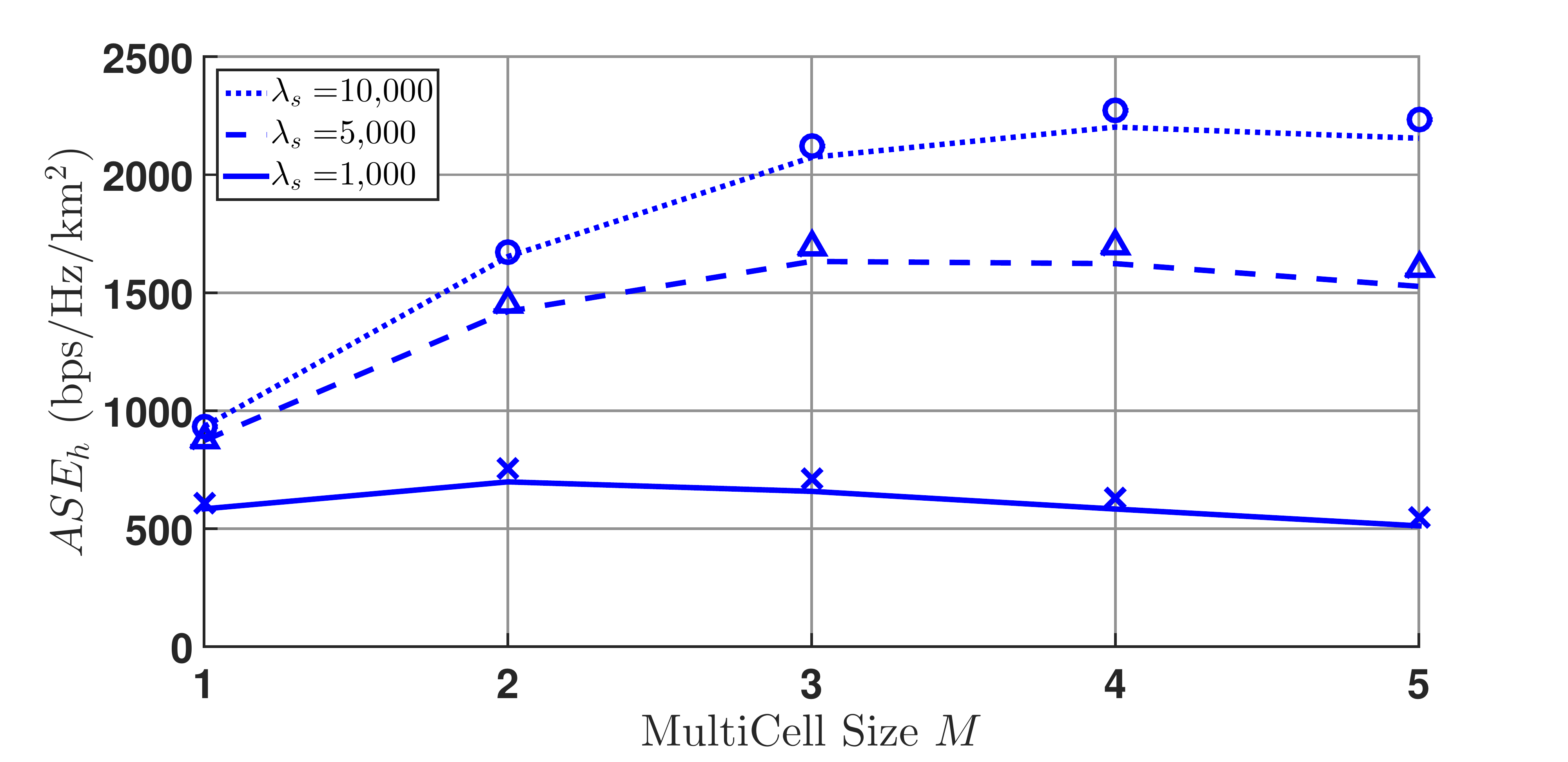}\label{fig:ASE_h_vs_M_b}}
    \caption{Achievable ASE for HTCUs versus \textit{MultiCell} size $M$ for different small cell densities and normalized backhaul capacities with $\lambda_h=500$ HTCUs/km$^2$. Lines indicate analysis while markers represent simulations.}
    \label{fig:ASE_h_vs_M}
\end{figure*}

The achievable ASE of HTCUs $\mathcal{T}_h$ versus the \textit{MultiCell} size $M$ is shown in Fig. \ref{fig:ASE_h_vs_M} via both simulations and analytical expression derived in (\ref{eq:ase_h}). In Fig. \ref{fig:ASE_h_vs_M_a}, we show how $M$ affects $\mathcal{T}_h$ under different normalized backhaul capacities. As evident from these results with no limitation on backhaul capacity (i.e., $\rho=\infty$), single association ($M=1$) represents the best alternative for HTC. This is due to the high gain of the radio channel between the HTCU and its nearest serving cell. However, for the practical case when the capacity of the backhaul links is limited, increasing $M$ can benefit the HTCU. One can notice from Fig.\ref{fig:ASE_h_vs_M} that different system parameters such as $\rho$ and $\lambda_s$ result in different values for the optimum \textit{MultiCell} size $M=M^*$ at which the ASE is maximized. Note that for $M<M^*$, the backhaul capacity limitations yield lower ASE $\mathcal{T}_h$. For $M>M^*$, the channel degradation with larger distance between the HTCU and the serving cell results in lower ASE since the achievable rate on the radio links can be supported by the backhaul links of the $M^*$ nearest cells.

\begin{figure}
	\centering
	\includegraphics[width=0.65\textwidth, height=0.45\textwidth]{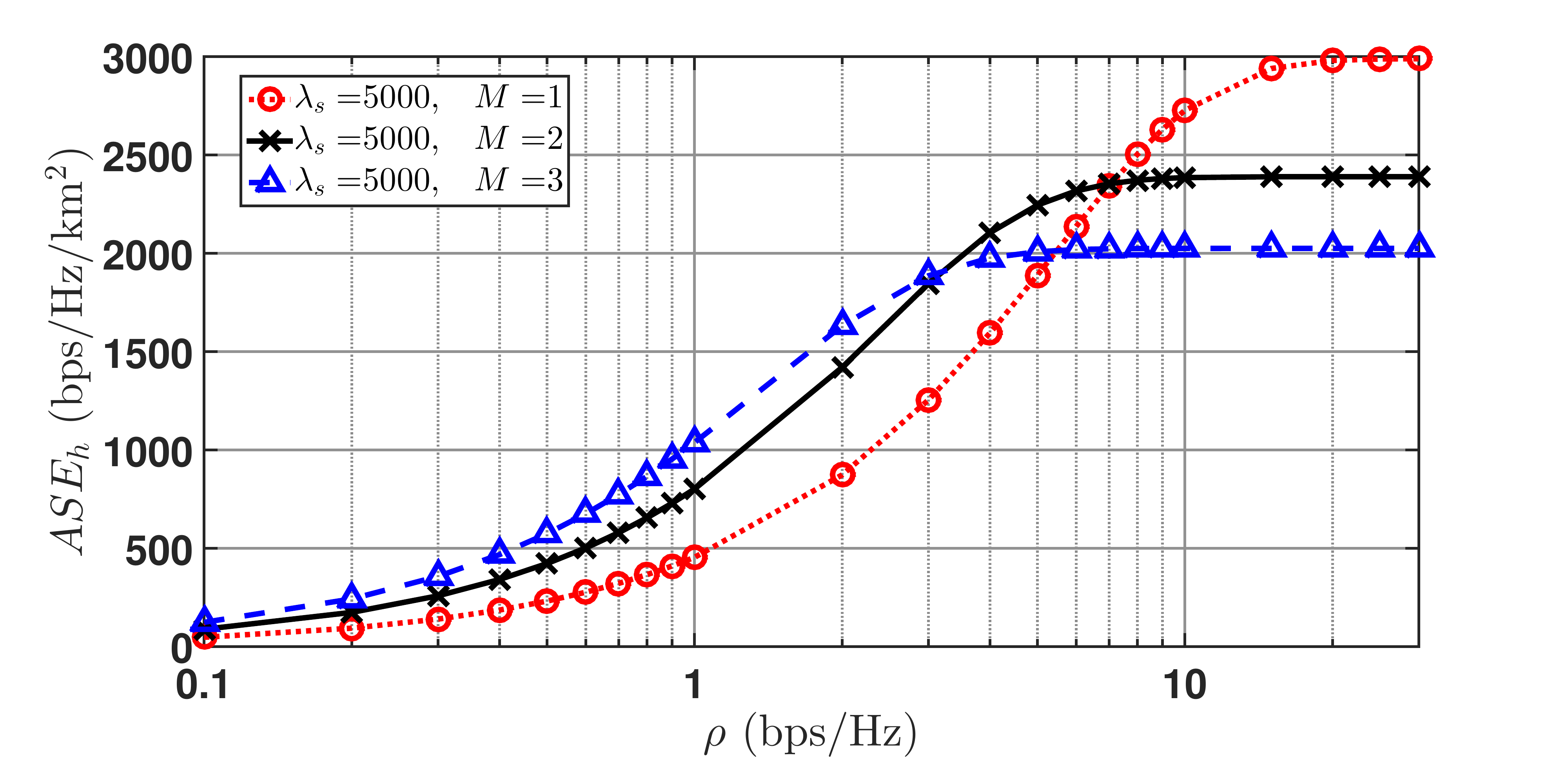}
	\caption{Achievable ASE for HTCUs versus normalized backhaul capacity $\rho$ for different \textit{MultiCell} sizes $M$, $\lambda_s=5,000$ SCs/km$^2$ and $\lambda_h=500$ HTCUs/km$^2$.}
	\label{fig:ASE_h_vs_Rho}
\end{figure}

Fig. \ref{fig:ASE_h_vs_Rho} shows the impact of increasing the backhaul capacity limit on the ASE for different \textit{MultiCell} sizes. For the specified system parameters of $\lambda_s$ and $\lambda_h$, we note that double association scheme ($M=2$) achieves the highest ASE for the HTCUs when the normalized backhaul capacity ranges from around $2$ to $8~bps/Hz$. However, if $\rho>8~bps/Hz$, single association becomes the best alternative.

\begin{figure}
	\centering
	\subfloat[$\lambda_h=500$]{\includegraphics[width=0.5\textwidth, height=0.4\textwidth]{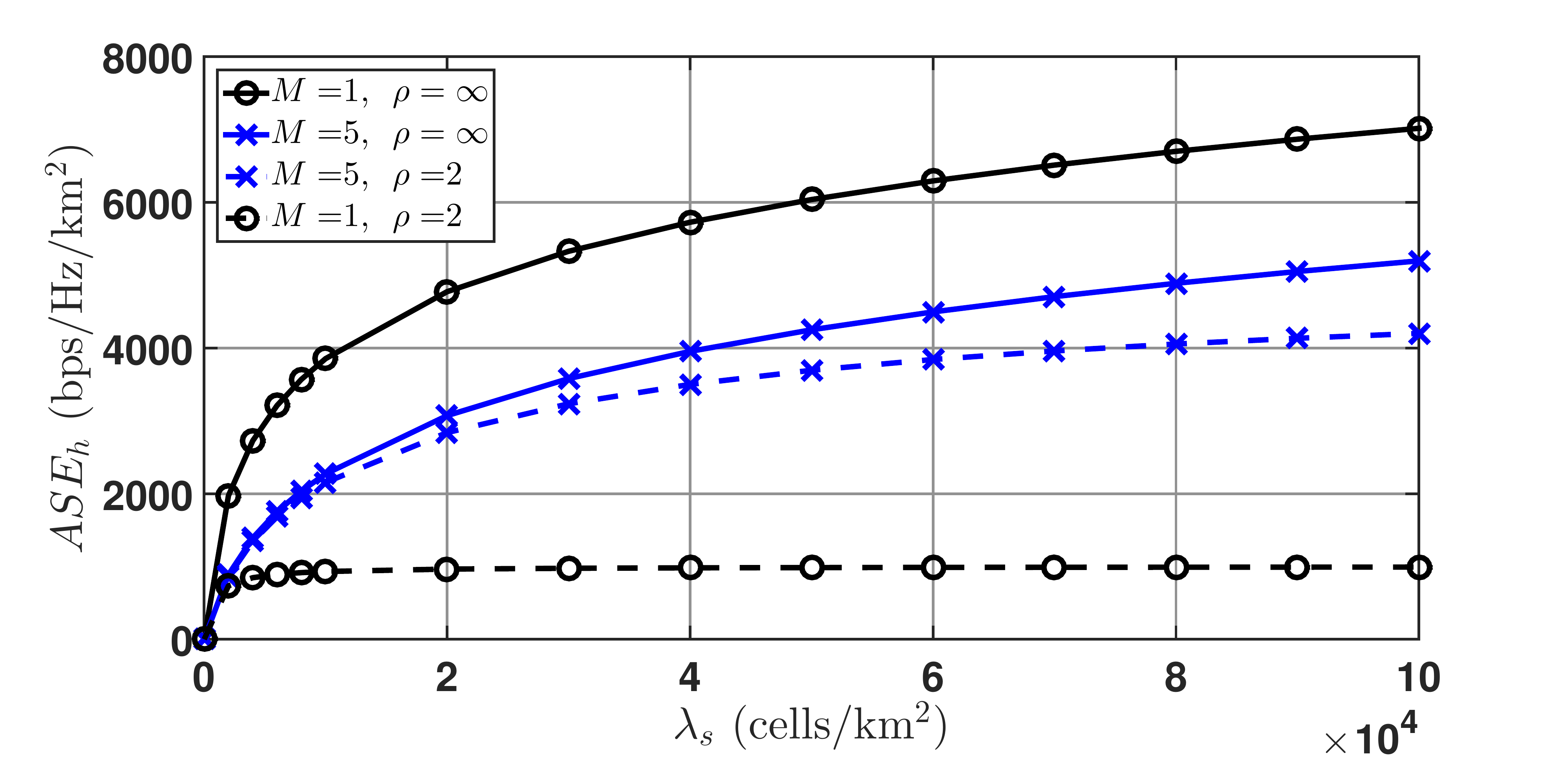}\label{fig:ASE_h_vs_la_s}}
	\subfloat[$\lambda_s=5,000$]{\includegraphics[width=0.5\textwidth, height=0.4\textwidth]{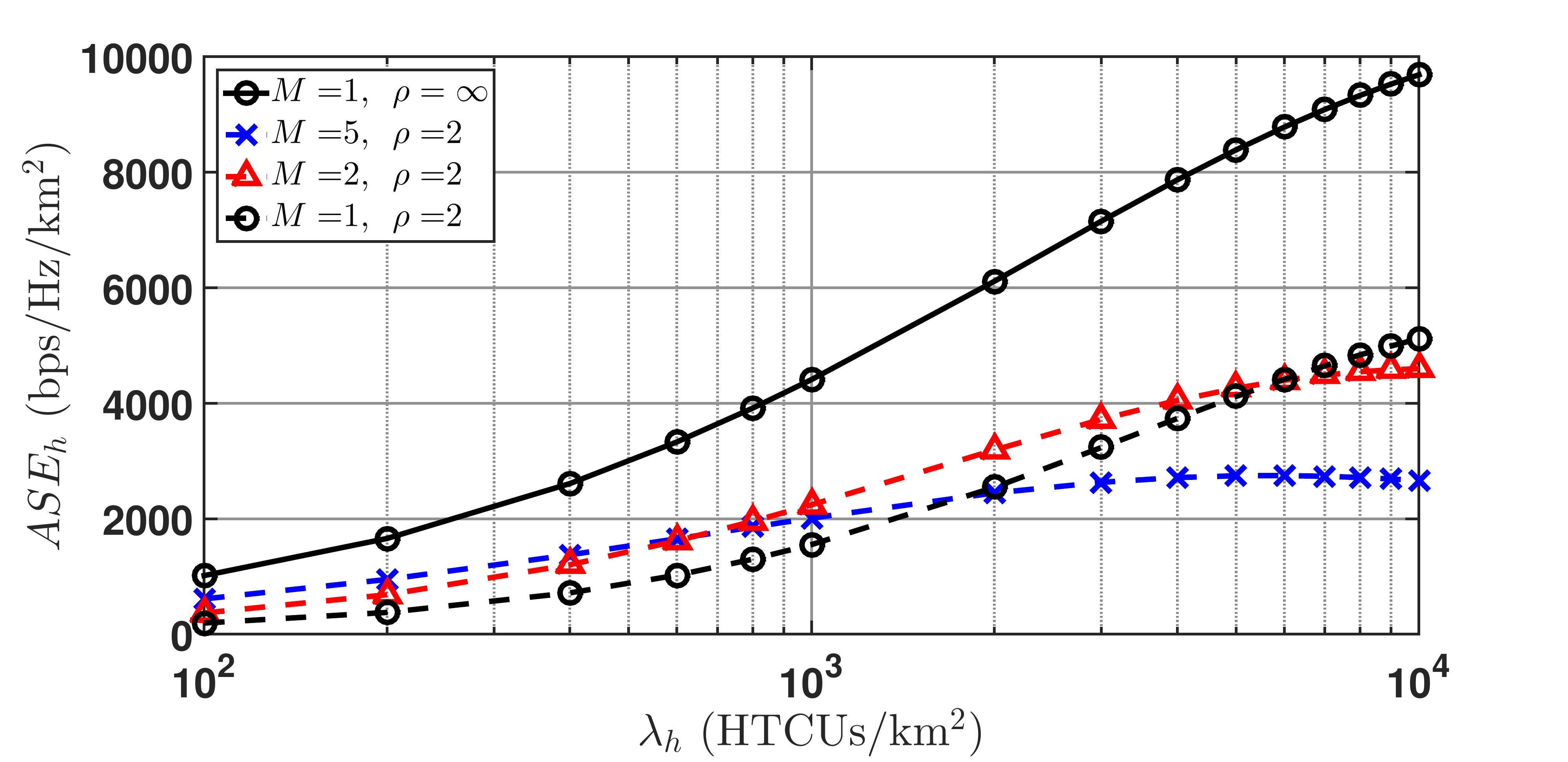}\label{fig:ASE_h_vs_la_h}}
	\caption{Achievable ASE for HTCUs versus small cell density $\lambda_s$ and HTCUs density $\lambda_h$ for different \textit{MultiCell} sizes $M$ and normalized backhaul capacities $\rho$.}
	\label{fig:ASE_h_vs_la_s_la_h}
\end{figure}

In Fig. \ref{fig:ASE_h_vs_la_s_la_h}, we show the impacts of varying the densities of both HTCUs and SCs on the performance of the HTC. One can notice that densifying the network enhances the achievable ASE for the HTCUs as shown in Fig. \ref{fig:ASE_h_vs_la_s}. However, when there exists backhaul capacity constraints, this achievable gain while adopting single association vanishes more rapidly as clear from the bottom curve.
In Fig. \ref{fig:ASE_h_vs_la_h}, we show that the achievable ASE for the HTCUs also increases with the density of HTCUs $\lambda_h$ under fixed density of small cells $\lambda_s=5000$ cells/km$^2$. However, when $\lambda_h$ further increases under backhaul capacity limitations, one can notice that the ASE starts to decrease. Such behaviour stems from the more disconnected links of higher-ordered tiers arousing from the fact that more users will be sharing the same cells as previously discussed in Fig. \ref{fig:double_assoc_scenarios}. Doing so leaves the connected tiers with only a portion of the available bandwidth compared to the whole bandwidth in case of single association. However, when the density of HTCUs $\lambda_h$ is relatively small compared to the density of small cells $\lambda_s$, multiple association improves the achievable ASE under the same backhaul capacity limitations.

\subsection{Machine-Type Communication}
In this subsection, we show the gains that can be attained by the MTCDs from adopting the proposed multiple association scheme. In Fig. \ref{fig:MTC_vs_M}, we show through both simulations and analyses, the impact of increasing the \textit{MultiCell} size $M$ on both performance metrics of MTC, namely, the density of supported MTCDs $\lambda_m^s$ and the achievable ASE $\mathcal{T}_m$. Clearly, one can notice that increasing $M$ yields significant gains in both $\lambda_m^s$ and $\mathcal{T}_m$ under different densities of small cells and HTCUs. Besides, one can notice from Fig. \ref{fig:ASE_m_vs_la_s} that the higher the \textit{MultiCell} size is, the higher the achievable gain is in the ASE of MTC with the increasing density of SCs $\lambda_s$. When single association is adopted, the ASE of MTC saturates more rapidly with $\lambda_s$. 

When the density of active MTCDs $\lambda_m^a$ is small, the achievable ASE is also small as shown in Fig. \ref{fig:ASE_m_vs_la_m}. This occurs since the available RBs by all active cells can support those active MTCDs while a plenty of resources are not being utilized. However, one can still notice that the achievable ASE is still higher with larger \textit{MultiCell} sizes $M$. When $M$ increases, more cells are activated which results in shorter distances between the active MTCDs and the SCs to which they are associated. Since all supported MTCDs transmit with fixed power, shortening the distances between the MTCDs and their serving SCs yields higher SIR and achievable ASE. As $\lambda_m^a$ increases, the ASE also increases until all the available RBs are utilized, then, the ASE saturates when no additional MTCDs can be supported. For higher $M$, there are more cells activated and the number of supported devices gets larger, hence, higher ASE is achieved.

\begin{figure*}
	\centering
	\subfloat[Density of supported MTCDs ($\lambda_m^s$)]{\includegraphics[width=0.5\textwidth, height=0.4\textwidth]{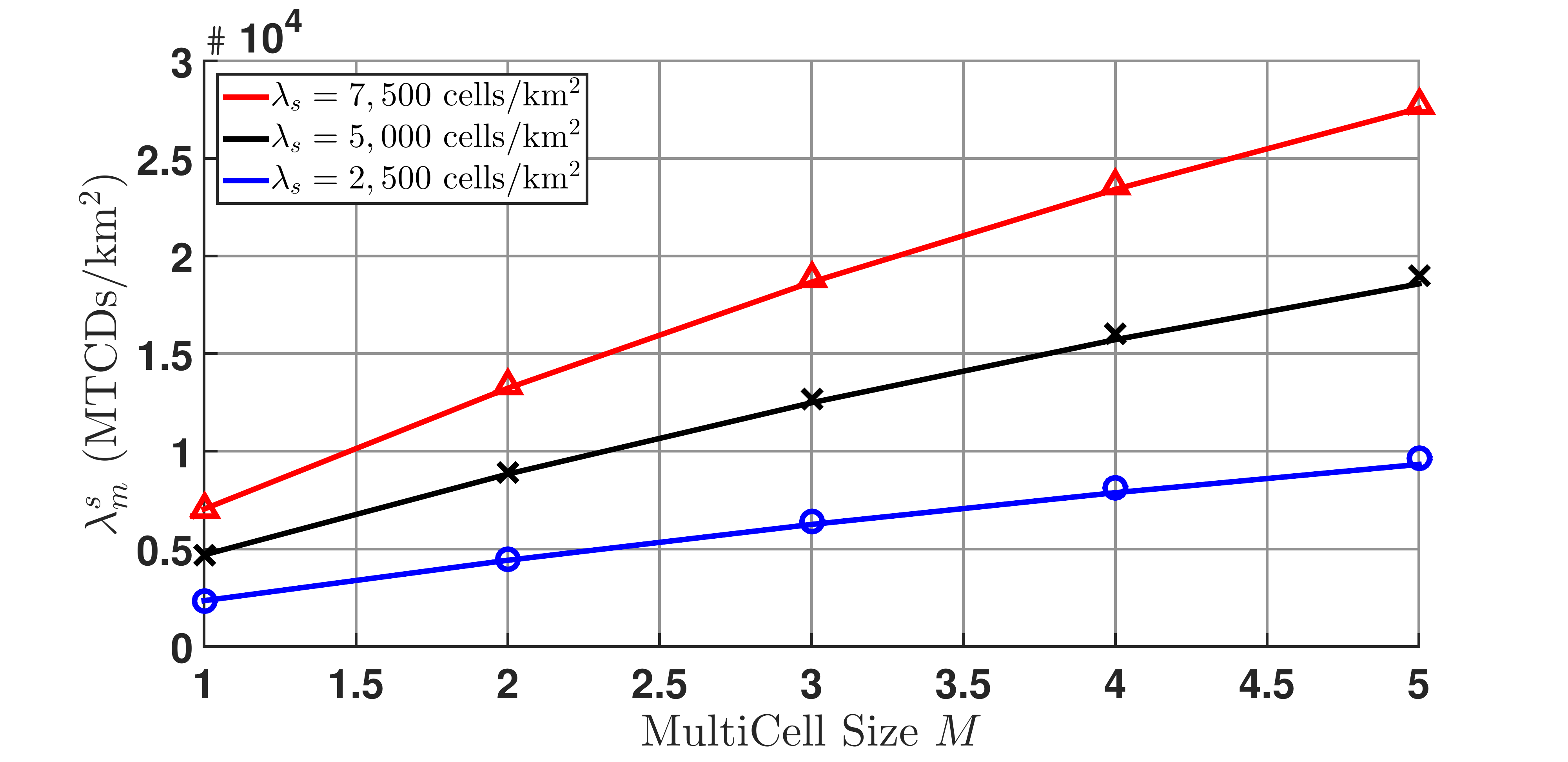}
	\label{fig:lambda_m_s_vs_M}}
    \subfloat[ASE of MTCDs ($\mathcal{T}_m$)]{\includegraphics[width=0.5\textwidth, height=0.4\textwidth]{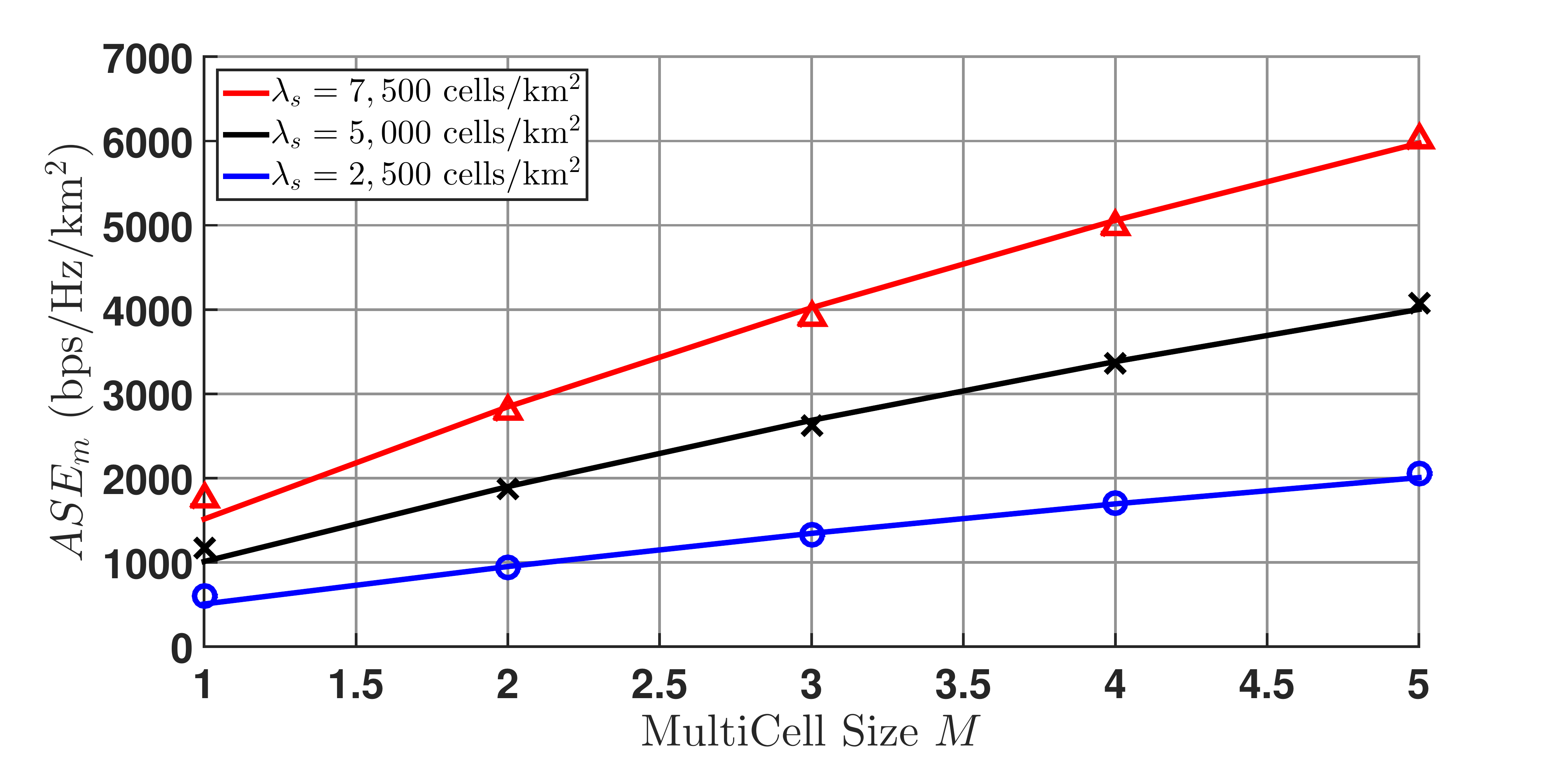}
	\label{fig:ASE_m_vs_M}}
	\caption{Density of supported MTCDs and ASE of MTC versus $M$ for different densities of SCs and HTCUs where $\lambda_s=10\lambda_h$, and active MTCDs density $\lambda_m^a=100,000$ (devices/km$^2$). Lines indicate analysis while markers represent simulations.}
	\label{fig:MTC_vs_M}
\end{figure*}

\begin{figure*}
	\centering
	\subfloat[$\lambda_m^a=100,000$ (devices/km$^2$)]{\includegraphics[width=0.5\textwidth, height=0.4\textwidth]{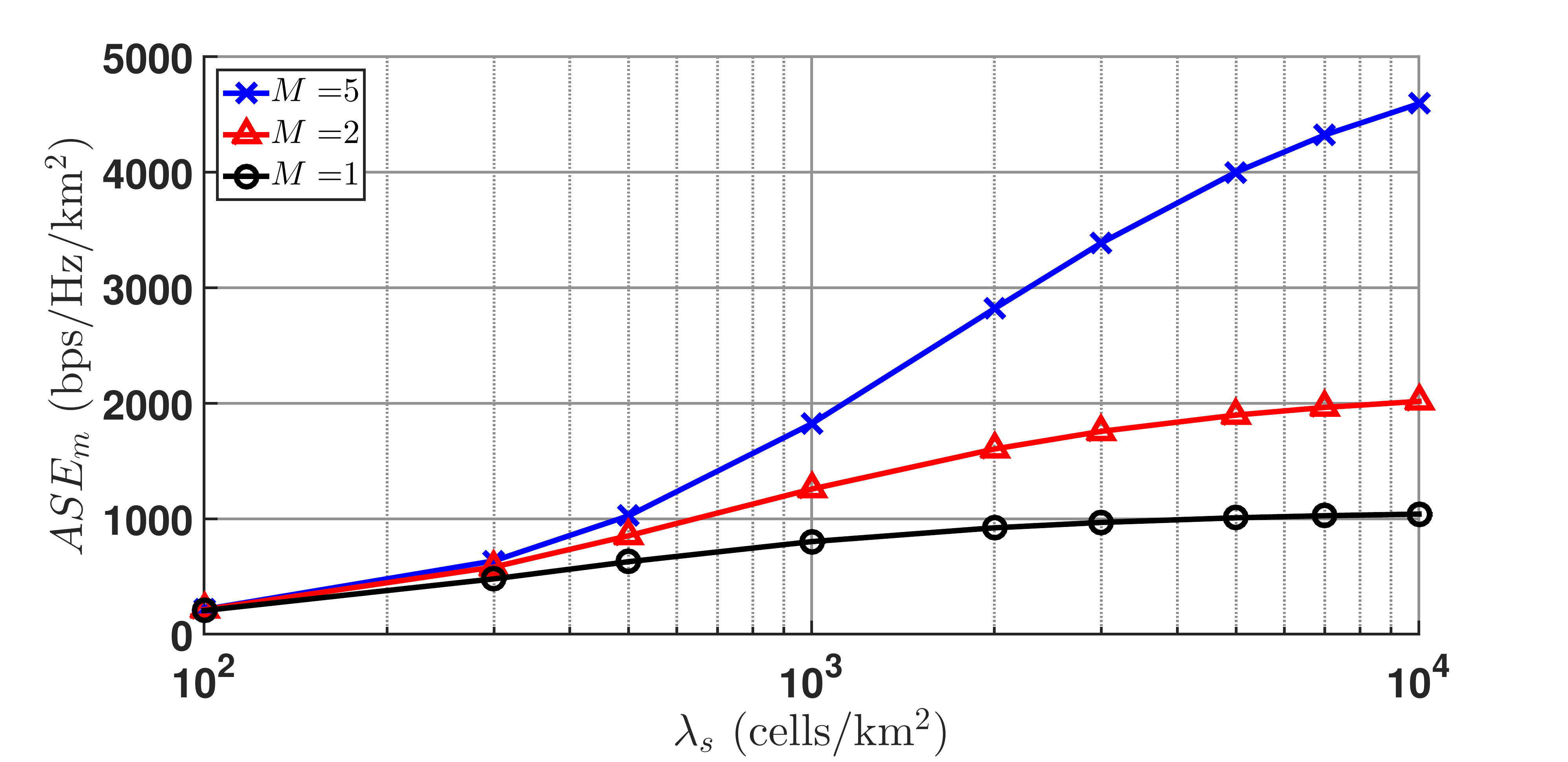}\label{fig:ASE_m_vs_la_s}}
	\subfloat[$\lambda_s=5,000$ (cells/km$^2$)]{\includegraphics[width=0.5\textwidth, height=0.4\textwidth]{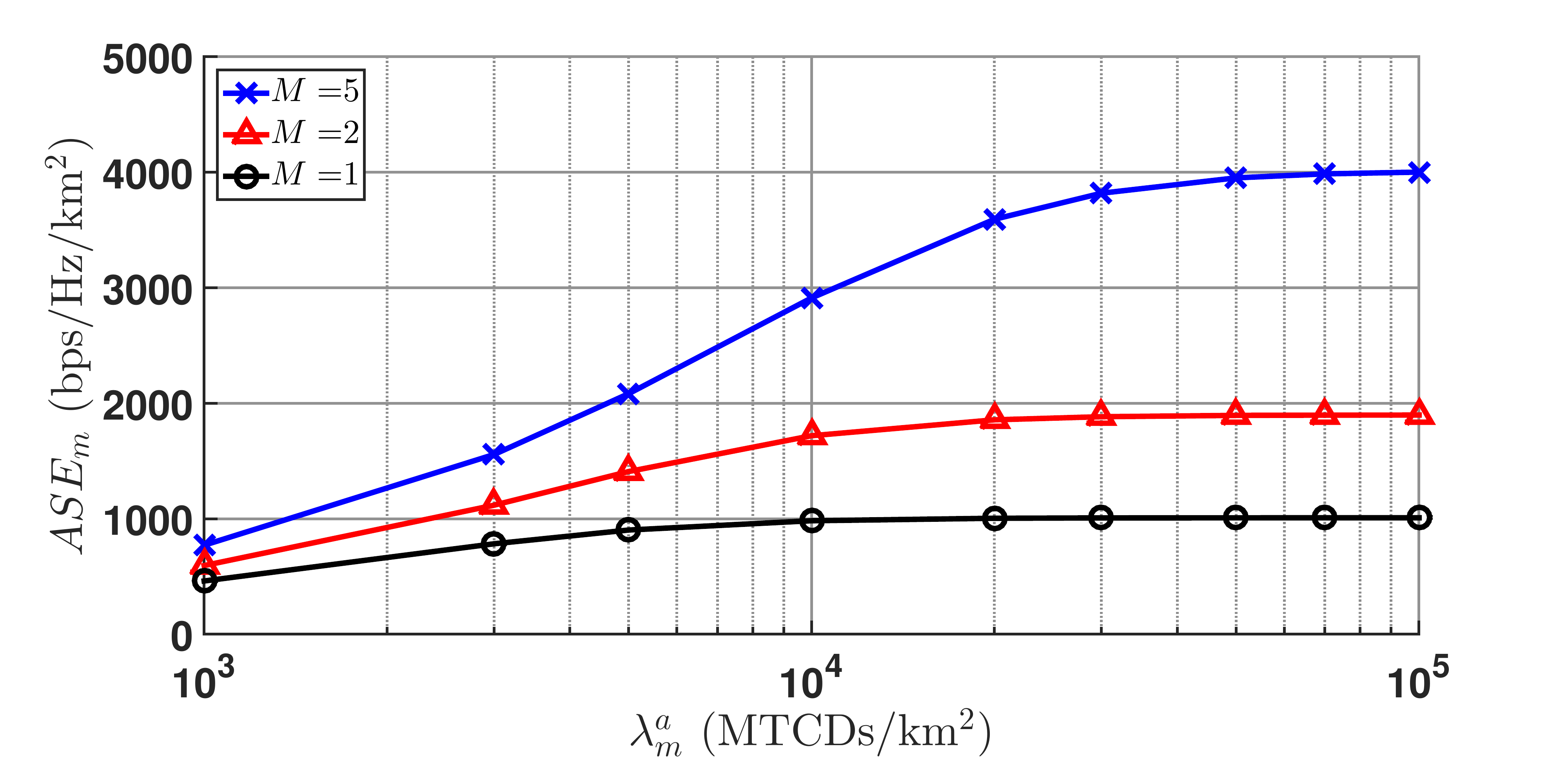}\label{fig:ASE_m_vs_la_m}}
	\caption{ASE of MTC versus small cell density $\lambda_s$ and active MTCDs density $\lambda_m^a$ for different \textit{MultiCell} sizes ($M$) with HTCUs density $\lambda_h=500$ (users/km$^2$).}
	\label{fig:MTC_vs_la_s_la_m}
\end{figure*}

\section{Conclusion}
We investigated the gains achieved from associating a Human-Type Communication User to multiple Small Cells in its vicinity. These gains are mainly obtained under limited backhaul capacity constraints that can be found in a UDN environment where fiber links require high deployment cost and time. Also, we show that the number of associated cells can be optimized for different system parameters such as backhaul capacities, cells and users densities.
In parallel, we investigated the gains achieved by the coexisting MTCDs. By adopting multiple association, more cells are activated and higher density of MTCDs can be supported. Following the proposed multiple association scheme, we showed how UDN can be useful in supporting the different use cases targeted in 5G and beyond such as Enhanced Mobile Broad-Band and massive Machine-Type Communication. In this regard, we derived a mathematical framework analysis using tools from stochastic geometry to obtain analytical expressions for the achievable Area Spectral Efficiency under backhaul capacity limitations. Further extensions to this investigation could be resource allocation optimization in terms of frequency, power, and \textit{MultiCell} size.

\bibliographystyle{IEEEtran}
\bibliography{IEEEabrv,references}
\end{document}